\documentclass[journal,onecolumn,11pt,twoside]{IEEEtran}
\usepackage{epsfig,graphics,subfigure}
\usepackage{amsmath,amssymb,amsthm,multirow,balance,dsfont,hyperref,cases}
\usepackage{cite,cleveref}
\usepackage[table]{xcolor}
\usepackage{multirow,array}
\usepackage[ruled,vlined]{algorithm2e}
\usepackage{colortbl}
\usepackage{tabularx}

\usepackage{scalerel,stackengine}
\stackMath
\newcommand\widecheck[1]{%
\savestack{\tmpbox}{\stretchto{%
  \scaleto{%
    \scalerel*[\widthof{\ensuremath{#1}}]{\kern-.3pt\bigwedge\kern-.3pt}%
    {\rule[-\textheight/3]{1ex}{\textheight}}
  }{\textheight}%
}{0.5ex}}%
\stackon[1pt]{#1}{\scalebox{-0.8}{\tmpbox}}%
}

\newcommand{\vertiii}[1]{{\left\vert\kern-0.25ex\left\vert\kern-0.25ex\left\vert #1 
    \right\vert\kern-0.25ex\right\vert\kern-0.25ex\right\vert}}
\newtheorem{theorem}{Theorem}
\newtheorem{lemma}{Lemma}

\newtheorem{assumption}{Assumption}

\usepackage{color}
\def\cblue{\textcolor{black}}
\def\cred{\textcolor{black}}
\def\cmag{\textcolor{black}}

\newcommand{\boldh}{\boldsymbol{h}}

\newcommand{\bs}{\boldsymbol{s}}
\newcommand{\bu}{\boldsymbol{u}}
\newcommand{\bv}{\boldsymbol{v}}
\newcommand{\bw}{\boldsymbol{w}}
\newcommand{\bz}{\boldsymbol{z}}
\newcommand{\bx}{\boldsymbol{x}}
\newcommand{\by}{\boldsymbol{y}}

\newcommand{\bpsi}{\boldsymbol{\psi}}

\newcommand{\bH}{\boldsymbol{H}}

\newcommand{\cA}{\mathcal{A}}
\newcommand{\cB}{\mathcal{B}}

\newcommand{\cd}{\scriptstyle\mathcal{D}}
\newcommand{\cD}{\mathcal{D}}
\newcommand{\cE}{\mathcal{E}}

\newcommand{\cR}{\mathcal{R}}

\newcommand{\cH}{\mathcal{H}}
\newcommand{\cI}{\mathcal{I}}
\newcommand{\cJ}{\mathcal{J}}

\newcommand{\cL}{\mathcal{L}}

\newcommand{\cN}{\mathcal{N}}
\newcommand{\cP}{\mathcal{P}}

\newcommand{\cU}{\mathcal{U}}
\newcommand{\cu}{{\scriptscriptstyle\mathcal{U}}}

\newcommand{\cV}{\mathcal{V}}

\newcommand{\cw}{{\scriptstyle\mathcal{W}}}

\newcommand{\ccw}{{\scriptscriptstyle\mathcal{W}}}

\newcommand{\bcB}{\boldsymbol{\cal{B}}}
\newcommand{\bcH}{\boldsymbol{\cal{H}}}

\newcommand{\bcD}{\boldsymbol{\cal{D}}}
\newcommand{\bcF}{\boldsymbol{\cal{F}}}
\newcommand{\bcw}{\boldsymbol{\cw}}

\newcommand{\bwt}{\widetilde \bw}

\newcommand{\bcwt}{\widetilde\bcw}
\newcommand{\bcwb}{\overline\bcw}
\newcommand{\bcwc}{\widecheck{\bcw}}

\newcommand{\cwb}{\overline{\cw}}
\newcommand{\wb}{\overline{w}}

\newcommand{\bsb}{\overline{\bs}}
\newcommand{\bsc}{\widecheck{\bs}}

\newcommand{\expec}{\mathbb{E}}

\newcommand{\col}{\text{col}}

\newcommand{\tr}{\text{Tr}}
\newcommand{\diag}{\text{diag}}

\DeclareMathOperator*{\minimize}{minimize}
\DeclareMathOperator*{\st}{subject~to}

\newcolumntype{C}[1]{>{\centering\arraybackslash}m{#1}}
\begin{document}
\title{Adaptation and learning over networks under  subspace constraints -- Part I: Stability Analysis}
\author{Roula Nassif$^\dag$, \IEEEmembership{Member, IEEE},  Stefan Vlaski$^{\dag,\ddag}$, \IEEEmembership{Student Member, IEEE}, \\
Ali H. Sayed$^\dag$, \IEEEmembership{Fellow Member, IEEE}\\
\thanks{This work was supported in part by NSF grant CCF-1524250. A short conference version of this work appears in~\cite{nassif2019distributed}.}

\vspace{0.5cm}
\small{\linespread{0.2}$^\dag$ Institute of Electrical Engineering, EPFL, Switzerland}\\
\vspace{0.1cm}
\small{\linespread{0.2}$^\ddag$ Electrical Engineering Department, UCLA, USA}\\
\vspace{0.3cm}
roula.nassif@epfl.ch\hspace{0.5cm} stefan.vlaski@epfl.ch\hspace{0.5cm} ali.sayed@epfl.ch}

%

\maketitle

\begin{abstract}
This paper considers optimization problems over networks where agents have individual objectives to meet, or individual parameter vectors to estimate, subject to subspace constraints that \cblue{require} the objectives across the network  to lie in  low-dimensional \cblue{subspaces}. This constrained formulation includes consensus optimization as a special case, and allows for more general task relatedness models such as smoothness. While such formulations can be  solved via projected gradient descent, the resulting algorithm is not distributed. \cblue{Starting from} the centralized solution, we propose an iterative and distributed implementation of the projection step, which runs in parallel with the stochastic gradient descent update. 
We establish in this Part I of the work that, for small step-sizes $\mu$, the proposed distributed adaptive strategy leads to small estimation errors on the order of $\mu$. \cmag{We examine in the accompanying Part II~\cite{nassif2019adaptation}} the steady-state performance. The results will reveal explicitly the influence of the gradient noise, data characteristics, and subspace constraints, on the network performance. The results will also show that in the small step-size regime, the iterates generated by the distributed algorithm achieve the centralized steady-state performance. 
\end{abstract}


\newpage

\section{Introduction}
Distributed inference allows a collection of interconnected agents to perform parameter estimation tasks from streaming data by relying solely on local computations and interactions with immediate neighbors. Most prior literature focuses on \textit{consensus} problems, where agents with separate objective functions need to agree on a common parameter vector corresponding to the minimizer of the aggregate sum of the individual costs, namely, 
\begin{equation}
\label{eq: consensus optimization problem}
w^o=\arg\min_{w}\sum_{k=1}^N J_k(w),
\end{equation}
where $J_k(\cdot)$ is the cost function at agent $k$, $N$ is the number of agents in the network, and $w\in\mathbb{C}^L$ is the global parameter vector, which all agents need to agree upon--see Fig.~\ref{fig: illustration of the concept} (middle). Each agent seeks to estimate $w^o$ through local computations and communications among neighboring agents without the need to know any of  the costs besides their own. Among many useful strategies that have been proposed in the literature~\cite{sayed2013diffusion,sayed2014adaptation,sayed2014adaptive,bertsekas1997new,nedic2009distributed,dimakis2010gossip,chouvardas2011adaptive,braca2008enforcing}, diffusion strategies~\cite{sayed2013diffusion,sayed2014adaptation,sayed2014adaptive} are particularly attractive since they are scalable, robust, and enable continuous learning and adaptation in response to drifts in the location of the minimizer.  

However, there exist many network applications that require more complex models and flexible algorithms than consensus implementations since their agents may involve the need to estimate and track multiple distinct, though related, objectives.  For instance, in distributed power system state estimation, the local state vectors to be estimated at neighboring control centers may overlap partially since the areas in a power system are interconnected~\cite{korres2011distributed,kekatos2013distributed}. Likewise, in monitoring applications,  agents need to track the movement of multiple correlated targets and to exploit the correlation profile in the data for enhanced accuracy~\cite{chen2014multitask,khawatmi2016decentralized}. Problems of this kind, where nodes need to infer multiple, though related, parameter vectors, are referred to as multitask problems. Existing strategies to address multitask problems generally exploit prior knowledge on how the tasks across the network relate to each other~\cite{plata2017heterogeneous,korres2011distributed,chen2014multitask,nassif2018distributed,koppel2016proximity,hallac2015network,cao2018distributed,nassif2019regularization,chen2014diffusion,mota2015distributed,alghunaim2017distributed,platachaves2015distributed,sahu2018cirfe,kekatos2013distributed,nassif2017diffusion,alghunaim2018dual,hua2017penalty,barbarossa2009distributed,khawatmi2016decentralized}. For example, one way to model relationships among tasks is to formulate convex optimization problems with appropriate co-regularizers between neighboring agents~\cite{chen2014multitask,nassif2018distributed,koppel2016proximity,hallac2015network,cao2018distributed}. Graph spectral regularization can also be used in order to leverage more thoroughly the graph spectral information and improve the multitask network performance~\cite{nassif2019regularization}. In other applications, it may happen that the parameter vectors to be estimated at neighboring agents are related according to a set of linear equality constraints~\cite{mota2015distributed,alghunaim2017distributed,platachaves2015distributed,chen2014diffusion,nassif2017diffusion,kekatos2013distributed,sahu2018cirfe}.

\cmag{However, in this paper,  and the accompanying Part II~\cite{nassif2019adaptation},} we consider multitask inference problems where each agent seeks to minimize an individual cost (expressed as the expectation of some loss function), and where the collection of parameter vectors to be estimated across the network is required to lie in a low-dimensional subspace--see Fig.~\ref{fig: illustration of the concept} (left).  That is, we let $w_k\in{\mathbb{C}^{M_k}}$ denote some parameter vector at node $k$ and let $\cw=\col\{w_1,\ldots,w_N\}$ denote the collection of parameter vectors from across the network. We associate with each agent $k$ a differentiable convex cost $J_k(w_k):{\mathbb{C}^{M_k}}\rightarrow\mathbb{R}$, which is expressed as the expectation of some loss function $Q_k(\cdot)$ and written as $J_k(w_k)=\expec Q_k(w_k;\bx_k)$, where $\bx_k$ denotes the random data. The expectation is computed over the distribution of the data. Let $M=\sum_{k=1}^NM_k$. We consider constrained problems of the  form: 
\begin{equation}
\label{eq: constrained optimization problem}
\begin{split}
\cw^o=&\,\arg\min_{\ccw}~~ J^\text{glob}(\cw)\triangleq\sum_{k=1}^N J_k(w_k),\\
&\,\st~~\cw\in\cR(\cU),
\end{split}
\end{equation}
where $\cR(\cdot)$ denotes the range space operator, and $\cU$ is an $M\times P$ full-column rank matrix with  $P\ll M$. 
Each agent $k$ is interested in estimating  the $k$-th $M_k\times 1$ subvector $w^o_k$ of $\cw^o=\col\{w^o_1,\ldots,w^o_N\}$. In order to solve~\eqref{eq: constrained optimization problem} iteratively, the gradient projection method can be applied~\cite{bertsekas1999nonlinear}: 
\begin{equation}
\label{eq: centralized solution}
\cw_i=\cP_{\cu}\left(\cw_{i-1}-\mu\,\col\left\{\nabla_{w_k^*}J_k(w_{k,i-1})\right\}_{k=1}^N\right),\quad i\geq 0,
\end{equation}
where \cblue{$\cw_i=\col\{w_{1,i},\ldots,w_{N,i}\}$ with $w_{k,i}$ the estimate of $w^o_k$ at iteration $i$ and agent $k$}, $\mu>0$ is a small step-size parameter, $\nabla_{w_k^*}J_k(\cdot)$ is the (Wirtinger) complex gradient~\cite[Appendix~A]{sayed2014adaptation} of $J_k(\cdot)$ relative to $w^*_k$ (complex conjugate of $w_k$), and \cblue{$\cP_{\cu}$ is the projector onto the $P$-dimensional subspace of $\mathbb{C}^M$ spanned by the columns of $\cU$:
\begin{equation}
\label{eq: P_R(U)}
\cP_{\cu}=\cU(\cU^*\cU)^{-1}\cU^*,
\end{equation} 
where we used the fact that $\cU$ is a full-column rank matrix.}

 We are particularly interested in solving the problem in the \textit{stochastic} setting when the distribution of the data $\bx_k$ is generally unknown. This means that the risks $J_k(\cdot)$ and their gradients $\nabla_{w_k^*}J_k(\cdot)$ are unknown. As such, approximate gradient vectors need to be employed. A common construction in stochastic approximation theory is to employ the following approximation at iteration $i$:
\begin{equation}
\label{eq: stochastic approximation of the gradient}
\widehat{\nabla_{w_k^*}J_k}(w_k)={\nabla_{w_k^*}Q_k}(w_k;\bx_{k,i}),
\end{equation}
where $\bx_{k,i}$ represents the data observed at iteration $i$. The difference between the true gradient and its approximation is called  \textit{gradient noise}. This noise will seep into the operation of the algorithm and one main challenge is to show that despite its presence, agent $k$ is \cblue{still able} to approach $w^o_k$ asymptotically. 

Although the gradient update in~\eqref{eq: centralized solution} \cblue{and~\eqref{eq: stochastic approximation of the gradient}} can be performed locally at agent $k$, the projection operation requires a fusion center. To see this, let us introduce an intermediate variable $\psi_{k,i}$ at node $k$:
\begin{equation}
\psi_{k,i}=w_{k,i-1}-\mu\nabla_{w_k^*}J_k(w_{k,i-1}).
\end{equation}
After evaluating $\psi_{k,i}$ locally, each agent at each iteration needs to send its estimate $\psi_{k,i}$ to a fusion center, which performs the projection operation in~\eqref{eq: centralized solution} by computing $\cw_i=\cP_{\cu}\col\{\psi_{1,i},\ldots,\psi_{N,i}\}$, and then sends the resulting estimates $w_{k,i}$ back to the agents. While centralized solutions can be powerful, decentralized solutions are more attractive since they are more robust and respect the privacy policy at each agent~\cite{sayed2014adaptation}. Thus, a second challenge we face in this paper is how to carry out the projection through a \textit{distributed} network where each node performs local computations and exchanges information only with its neighbors. 

We  propose in Section~\ref{sec: Distributed inference under subspace constraints} an adaptive and distributed iterative algorithm allowing each agent $k$ to converge, in the mean-square-error sense, within $O(\mu)$ from the solution $w^o_k$ of~\eqref{eq: constrained optimization problem}, for sufficiently small $\mu$. Conditions on the network topology and signal subspace ensuring the feasibility of a distributed implementation are provided. We also show how some well-known network optimization problems, such as consensus optimization~\cite{sayed2013diffusion,sayed2014adaptation,sayed2014adaptive} and multitask smooth optimization~\cite{nassif2018distributed}, can be recast in the form~\eqref{eq: constrained optimization problem} and addressed with the strategy proposed in this paper. \cmag{The analysis in Section~\ref{sec: Stochastic performance analysis} of this Part~I shows that, for sufficiently small $\mu$, the proposed adaptive strategy leads to small estimation errors on the order of the small step-size. Building on the results of this Part~I, we shall derive in  Part~II~\cite{nassif2019adaptation} a closed-form expression for the steady-state network mean-square-error performance.} This closed form expression will reveal explicitly the influence of the data characteristics (captured by the second-order properties of the costs and second-order moments of the gradient noises) and subspace constraints (captured by $\cU$), on the network performance. The results will also show that, in the small step-size regime, the iterates generated by the distributed implementation achieve the centralized steady-state performance. 
\cmag{For illustration purposes, distributed sub-optimal beamforming is considered in Section~\ref{subsec: Distributed linearly-constrained-minimum-variance (LCMV) beamformer} \cmag{of this Part~I}. }

\noindent\textbf{Notation:} All vectors are column vectors. Random quantities are denoted in boldface. Matrices are denoted in capital letters while vectors and scalars are denoted in lower-case letters. We use the symbol $(\cdot)^\top$ to denote matrix transpose, the symbol $(\cdot)^*$ to denote matrix complex-conjugate transpose, and the symbol $\tr(\cdot)$ to denote trace operator. The symbol $\diag\{\cdot\}$ forms a matrix from block arguments by placing each block immediately below and to the right of its predecessor. The operator $\col\{\cdot\}$ stacks the column vector entries on top of each other. \cmag{The symbol $\otimes$ denotes} the Kronecker product. 
 The $M\times M$ identity matrix is denoted by~$I_M$.

\section{Distributed inference under subspace constraints}
\label{sec: Distributed inference under subspace constraints}
We move on to propose and study a distributed solution for solving~\eqref{eq: constrained optimization problem} with a continuous adaptation mechanism. The solution must rely on local computations and communications with immediate neighborhood, and operate in real-time on streaming data. To proceed with the analysis, one of the challenges we face is that the projection in~\eqref{eq: centralized solution} requires non-local exchange of information. Our strategy is to replace the $M\times M$ projection matrix $\cP_{\cu}$ in~\eqref{eq: centralized solution} by an $M\times M$ matrix $\cA$ that satisfies the following conditions:
\begin{numcases}{}
  \lim\limits_{i\rightarrow\infty}\cA^i=\cP_{\cu}, & \label{eq: condition 1 on A}\\
  A_{k\ell}=[\cA]_{k\ell}=0, \quad\text{if }\ell\notin\cN_k \text{ and } k\neq\ell,& \label{eq: condition 2 on A}
\end{numcases}
where $[\cA]_{k\ell}$ denotes the $(k,\ell)$-th block of $\cA$ of dimension $M_k\times M_{\ell}$ and $\cN_k$ denotes the neighborhood of agent $k$, i.e., the set of nodes connected to agent $k$ by an edge. The sparsity condition~\eqref{eq: condition 2 on A} characterizes the network topology and ensures local exchange of information at each time instant $i$. By replacing the projector $\cP_{\cu}$ in~\eqref{eq: centralized solution} by $\cA$ and the true gradients $\nabla_{w_k^*}J_k(\cdot)$ by their stochastic approximations, we obtain the following distributed adaptive solution at each agent~$k$:
\begin{equation}
\label{eq: distributed solution}
\left\lbrace
\begin{array}{rl}
\bpsi_{k,i}=&\hspace{-2mm}\bw_{k,i-1}-\mu\widehat{\nabla_{w_k^*}J_k}(\bw_{k,i-1}),\\
\bw_{k,i}=&\hspace{-2mm}\sum\limits_{\ell\in\cN_k}A_{k\ell}\bpsi_{\ell,i},
\end{array}
\right.
\end{equation}
where we used condition~\eqref{eq: condition 2 on A}, and where $\bpsi_{k,i}$ is an intermediate estimate and $\bw_{k,i}$ is the estimate of $w^o_k$ at agent $k$ and iteration $i$. As we shall see in Section~\ref{sec: Stochastic performance analysis}, condition~\eqref{eq: condition 1 on A} helps \cblue{ensure} convergence toward the optimum. Necessary and sufficient conditions for the matrix equation~\eqref{eq: condition 1 on A} to hold are given in the following lemma.
\begin{lemma}{\emph{(Necessary and sufficient conditions for~\eqref{eq: condition 1 on A})}} 
\label{lemm: matrix equation proposition}
The matrix equation~\eqref{eq: condition 1 on A} holds, if and only if, the following conditions on the projector $\cP_{\cu}$ and the matrix $\cA$ are satisfied:
\begin{align}
\cA\cP_{\cu}=\cP_{\cu},\label{eq: first condition in the proposition}\\
\cP_{\cu}\cA=\cP_{\cu},\label{eq: second condition in the proposition}\\
\rho(\cA-\cP_{\cu})<1,\label{eq: third condition in the proposition}
\end{align}
where $\rho(\cdot)$ denotes the spectral radius of its matrix argument. It follows that any $\cA$ satisfying condition~\eqref{eq: condition 1 on A} has  one as an eigenvalue with multiplicity $P$, and all other eigenvalues are strictly less than one in magnitude.
\end{lemma}
\begin{proof}
See \cred{Appendix~\ref{app: proof of lemma 1}}. The arguments are along the lines developed in~\cite{xiao2004fast} for distributed averaging with proper adjustments to handle general subspace constraints. 
\end{proof}
Note that conditions~\eqref{eq: first condition in the proposition}--\eqref{eq: third condition in the proposition} appeared previously   (with proof omitted) in the context of  distributed denoising in wireless sensor networks~\cite{barbarossa2009distributed}. In such problems, the $N$ sensors are observing $N$-dimensional signal, with each entry of the signal corresponding to one sensor. Using the prior knowledge that the observed signal belongs to a low-dimensional subspace, the sensor task is to denoise the corresponding entry of the signal by projecting  in a distributed iterative manner onto the signal subspace in order to improve the error variance. However, in this work, we consider the more general problem of distributed inference over networks.

\cblue{If we replace $\cP_{\cu}$ by~\eqref{eq: P_R(U)}, multiply both sides of~\eqref{eq: first condition in the proposition} by $\cU$, and multiply both sides of~\eqref{eq: second condition in the proposition} by $\cU^*$, conditions~\eqref{eq: first condition in the proposition} and~\eqref{eq: second condition in the proposition} reduce to:}
\begin{align}
\cA \,\cU&=\cU,\label{eq: first condition on eigenvector}\\
\cU^*\cA&=\cU^*.\label{eq: second condition on eigenvector}
\end{align}
Conditions~\eqref{eq: first condition on eigenvector} and~\eqref{eq: second condition on eigenvector} state that the $P$ columns of $\cU$ are right and left eigenvectors of $\cA$ associated with the eigenvalue $1$. Together with these two conditions, condition~\eqref{eq: third condition in the proposition} means that $\cA$ has $P$ eigenvalues at one, and that all other eigenvalues are strictly less than one in magnitude.
\begin{figure*}
\centering
\includegraphics[scale=0.33]{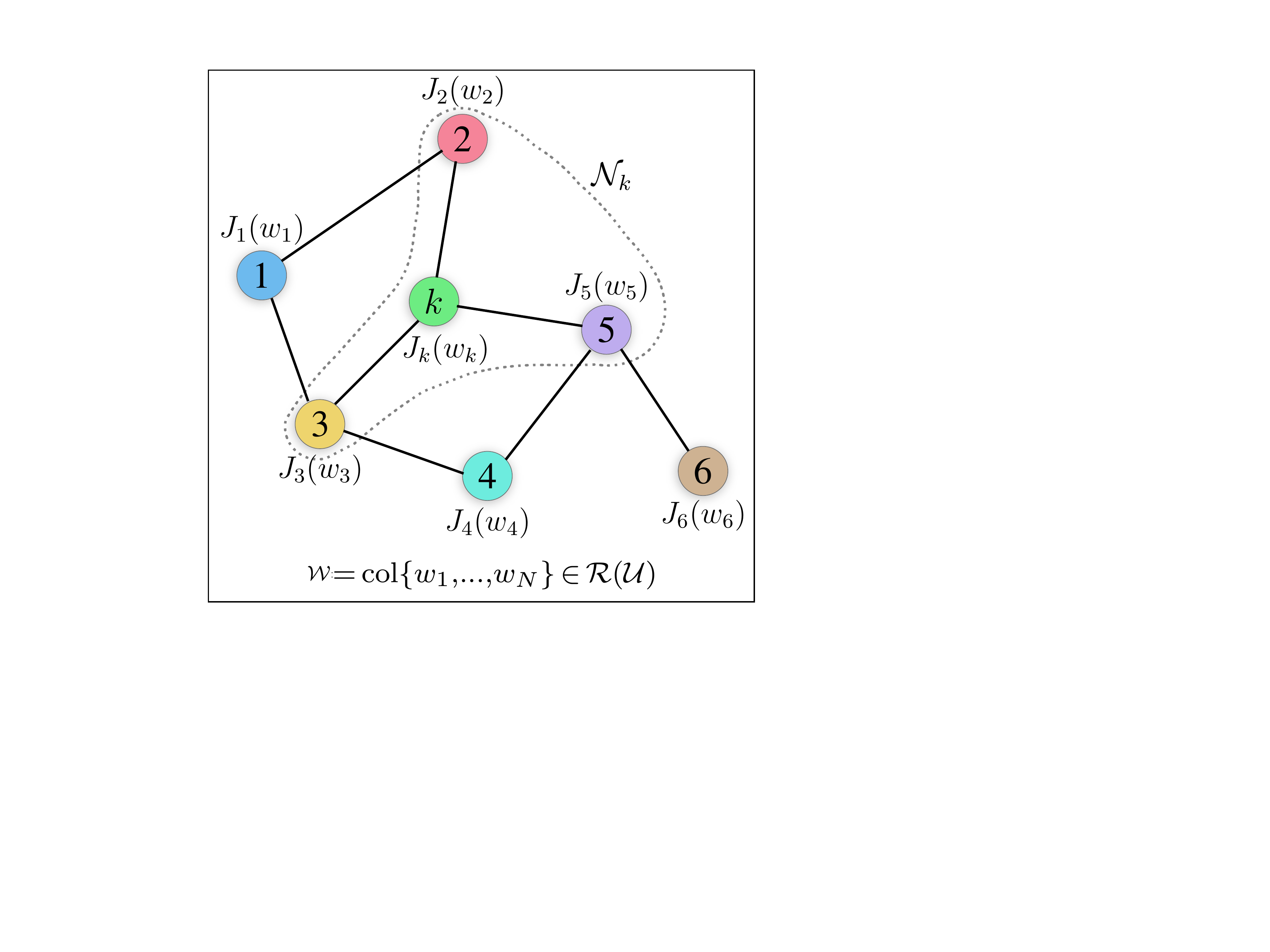}\quad
\includegraphics[scale=0.33]{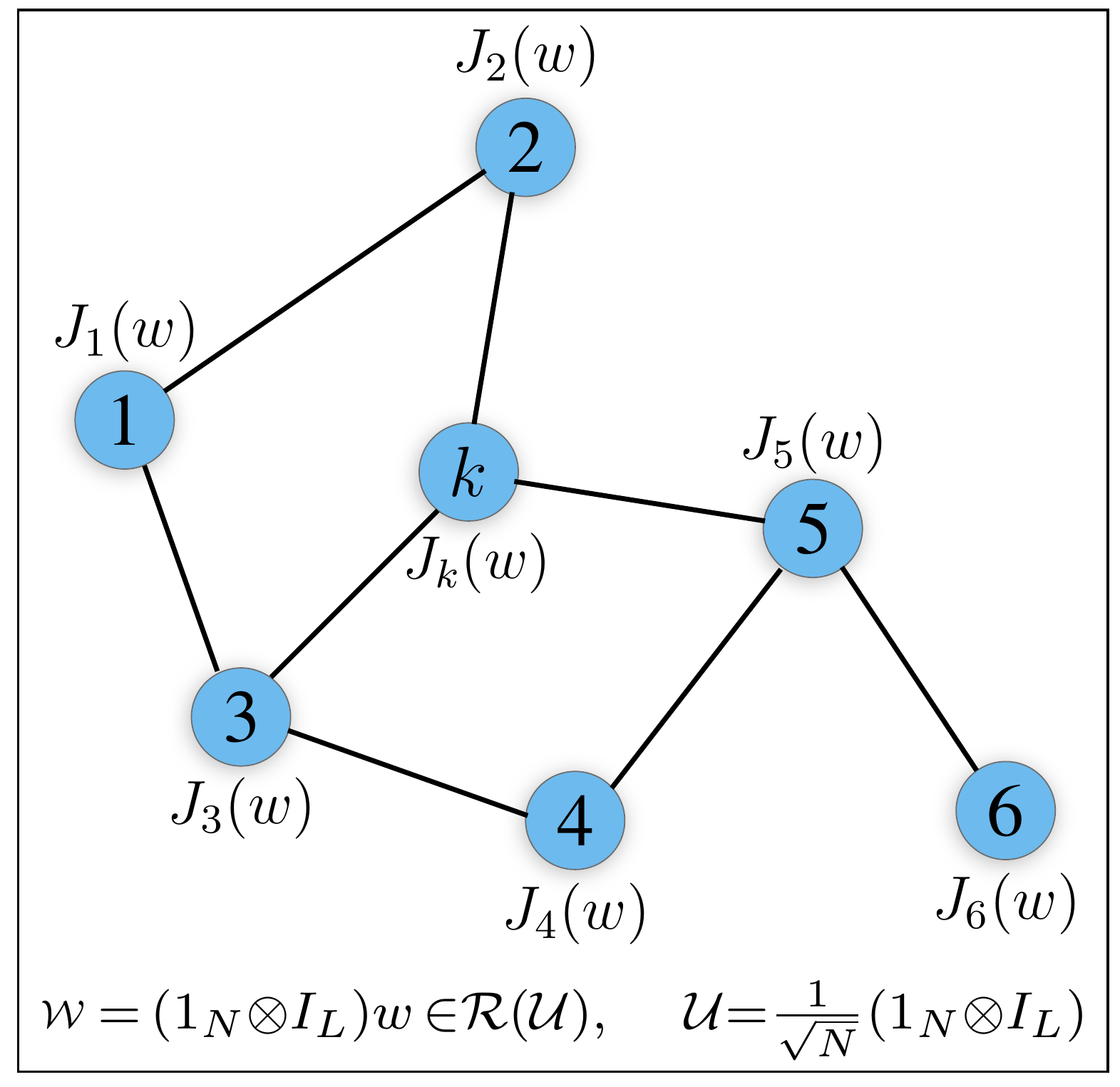}\quad
\includegraphics[scale=0.33]{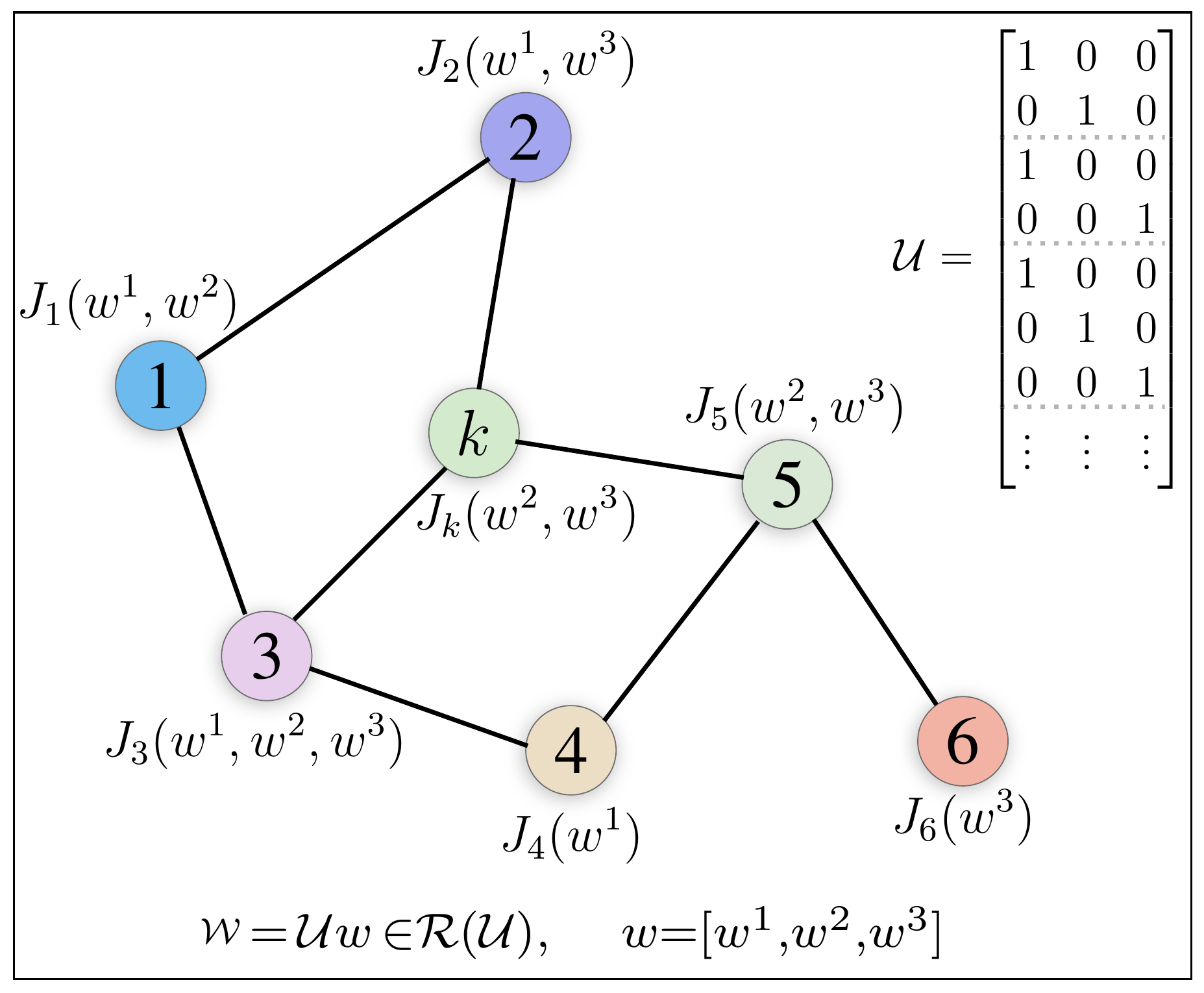}
\caption{Inference under subspace constraints. \textit{(Left)} Illustrative scheme of problem~\eqref{eq: constrained optimization problem}: Each agent $k$ in the network has an individual  $w_k$ to estimate, subject to subspace constraints that enforce the objectives across the network to lie in $\cR(\cU)$. \textit{(Middle)} Consensus optimization~\eqref{eq: consensus optimization problem}: Agents in the network seek to estimate an $L\times 1$ common vector $w$ corresponding to the minimizer of the aggregate sum of individual costs. \textit{(Right)} Coupled optimization~\cite{alghunaim2017distributed}: Different agents generally seek to estimate different, but overlapping, parameter vectors.}
\label{fig: illustration of the concept}
\end{figure*}
In the following, we discuss how some well-known network optimization problems can be recast in the form~\eqref{eq: constrained optimization problem} and addressed with strategies in the form of~\eqref{eq: distributed solution}.

\noindent\textbf{Remark 1.} (Distributed consensus optimization). Let $M_k=L$ for all agents. If we set in~\eqref{eq: constrained optimization problem} $P=L$ and $\cU=\frac{1}{\sqrt{N}}(\mathds{1}_N\otimes I_L)$ where $\mathds{1}_N$ is the $N\times 1$ vector of all ones, then solving problem~\eqref{eq: constrained optimization problem} will be equivalent to solving the well-known consensus problem~\eqref{eq: consensus optimization problem}. Different algorithms for solving~\eqref{eq: consensus optimization problem} over strongly-connected networks have been proposed~\cite{sayed2013diffusion,sayed2014adaptation,sayed2014adaptive,nedic2009distributed,bertsekas1997new,dimakis2010gossip,chouvardas2011adaptive}. By picking any $N\times N$ doubly-stochastic matrix $A=[a_{k\ell}]$ satisfying:
\begin{equation}
\label{eq: conditions on A for diffusion}
a_{k\ell}\geq 0, A\mathds{1}_N=\mathds{1}_N,\mathds{1}^\top_NA=\mathds{1}^\top_N, a_{k\ell}= 0~\text{if }\ell\notin\cN_k\text{ and } k\neq\ell
\end{equation} 
the diffusion strategy for instance takes the form~\cite{sayed2013diffusion,sayed2014adaptation,sayed2014adaptive}:
\begin{equation}
\label{eq: diffusion strategy}
\left\lbrace
\begin{array}{rl}
\bpsi_{k,i}=&\hspace{-2mm}\bw_{k,i-1}-\mu\widehat{\nabla_{w^*_k}J_k}(\bw_{k,i-1}),\\
\bw_{k,i}=&\hspace{-2mm}\sum\limits_{\ell\in\cN_k}a_{k\ell}\bpsi_{\ell,i}.
\end{array}
\right.
\end{equation}
Observe that this strategy  can be written in the form of~\eqref{eq: distributed solution} with $A_{k\ell}=a_{k\ell} I_L$ and $\cA=A\otimes I_L$. It can be verified that, when $A$ satisfies~\eqref{eq: conditions on A for diffusion} over a strongly connected network, the matrix $\cA$ will satisfy~\eqref{eq: condition 2 on A},~\eqref{eq: first condition on eigenvector},~\eqref{eq: second condition on eigenvector}, and~\eqref{eq: third condition in the proposition}. 
\qed

\noindent\textbf{Remark 2.} (Distributed coupled optimization). Similarly, with a proper selection of $\cU$, multitask inference problems with overlapping parameter vectors~\cite{alghunaim2017distributed,mota2015distributed,platachaves2015distributed} can also be recast in the form~\eqref{eq: constrained optimization problem}. This scenario is illustrated in Fig.~\ref{fig: illustration of the concept} (right). In this example, agent $k$ is influenced by only a subset of the entries of a global $w=[w^1,w^2,w^3]$ and seeks to estimate $w_k=[w^2,w^3]$.  For a given variable $w^\ell$ and any two arbitrary agents containing $w^\ell$ in their costs, it is assumed that the network topology is such that there exists at least one path linking one agent to the other~\cite{alghunaim2017distributed}. By properly selecting the matrix $\cU$, the network vector $\cw=\col\{w_1,\ldots,w_N\}$ can be written as $\cw=\cU w$ and, therefore, distributed coupled optimization can be recast in the form~\eqref{eq: constrained optimization problem}. It can be verified that the coupled diffusion strategy proposed in~\cite{alghunaim2017distributed} for solving this problem can be written in the form of~\eqref{eq: distributed solution} and that the (doubly-stochastic) matrix $\cA$ in~\cite{alghunaim2017distributed} satisfies conditions~\eqref{eq: condition 1 on A} and~\eqref{eq: condition 2 on A}. \qed

\noindent\textbf{Remark 3.} (Distributed optimization under affine constraints). Several existing works consider (distributed or offline) variations of the following problem~\cite{alghunaim2018dual,nassif2017diffusion,hua2017penalty}:
\begin{equation}
\label{eq: affine constrained optimization problem}
\begin{split}
\cw^o=&\,\arg\min_{\ccw}\quad\sum_{k=1}^N J_k(w_k),\\
&\,\st~~\cD^*\cw=d,
\end{split}
\end{equation}
where $\cD$ is an $M\times (M-P)$ full-column rank matrix and $d$ is an $(M-P)\times 1$ column vector.  It turns out that the online distributed strategy proposed in this work can be used to solve~\eqref{eq: affine constrained optimization problem} for general constraints that are not necessarily local. To see this, we first note that the gradient projection method can be applied to solve~\eqref{eq: affine constrained optimization problem}~\cite{bertsekas1999nonlinear}:
\begin{equation}
\label{eq: centralized solution affine constrained optimization problem}
\cw_i=\cP_{\cd}\left(\cw_{i-1}-\mu\,\col\left\{\nabla_{w^*_k}J_k(w_{k,i-1})\right\}_{k=1}^N\right)+d_{\cd},~~ i\geq 0,
\end{equation}
where $\cP_{\cd}\triangleq I_M-\cD(\cD^*\cD)^{-1}\cD^*$ and $d_{\cd}\triangleq\cD(\cD^*\cD)^{-1}d$. Since  $\cP_{\cd}$ is a projection matrix, it can be decomposed as $\cP_{\cd}=\sum_{p=1}^Pu_pu^*_p$ with $\{u_p\}$ the orthonormal eigenvectors \cblue{of $\cP_{\cd}$} associated with the $P$ eigenvalues at one, and thus $\cP_d$ can be replaced by $\cP_{\cu}=\cU\cU^*$ with $\cU=[u_1,\ldots, u_P]$. Therefore, solution~\eqref{eq: centralized solution affine constrained optimization problem} has a form similar to the earlier solution~\eqref{eq: centralized solution} with the rightmost term $d_{\cd}$ in~\eqref{eq: centralized solution affine constrained optimization problem} absent from~\eqref{eq: centralized solution}. Following the same line of reasoning that led to~\eqref{eq: distributed solution}, we can similarly obtain the following distributed adaptive solution for solving~\eqref{eq: affine constrained optimization problem}:
 \begin{equation}
\label{eq: distributed solution affine constrained optimization problem}
\left\lbrace
\begin{array}{rl}
\bpsi_{k,i}=&\hspace{-2mm}\bw_{k,i-1}-\mu\widehat{\nabla_{w^*_k}J_k}(\bw_{k,i-1}),\\
\bw_{k,i}=&\hspace{-2mm}\sum\limits_{\ell\in\cN_k}A_{k\ell}\bpsi_{\ell,i}+d_{\cD,k},
\end{array}
\right.
\end{equation}
where $d_{\cD,k}$ is the $k$-th sub-vector of $(I_M-\cA)d_{\cD}$ corresponding to node $k$ (see Appendix~\ref{app: affine constraints driving term}), and $\cA=[A_{k\ell}]$ is a properly selected matrix satisfying conditions~\eqref{eq: condition 1 on A} and~\eqref{eq: condition 2 on A}. Although algorithm~\eqref{eq: distributed solution affine constrained optimization problem} is different than~\eqref{eq: distributed solution} due to the presence of the constant term $d_{\cD,k}$, the mean-square-error analyzes of both algorithms are the same, as we shall see in Section~\ref{sec: Stochastic performance analysis}. In Section~\ref{subsec: Distributed linearly-constrained-minimum-variance (LCMV) beamformer}, we shall apply~\eqref{eq: distributed solution affine constrained optimization problem} to solve linearly constrained beamforming~\cite{trees2002optimum,frost1972algorithm}.\qed

\noindent\textbf{Remark 4.} (Distributed inference under smoothness). Let $M_k=L$ for all agents. In such problems, each agent $k$ in the network has an individual cost $J_k(w_k)$ to minimize subject to a smoothness condition over the graph. The smoothness requirement softens the transition in the tasks $\{w_k\}$ among neighboring nodes and can be measured in terms of a quadratic form of the graph Laplacian~\cite{nassif2018distributed}:
\begin{equation}
S(\cw)=\cw^\top\cL_c\cw=\frac{1}{2}\sum_{k=1}^N\sum_{\ell\in\cN_k}c_{k\ell}\|w_k-w_\ell\|^2,
\end{equation}
where $\cL_c=L_c\otimes I_L$ with $L_c=\diag\{C\mathds{1}_N\}-C$ denoting the graph Laplacian. The matrix $C=[c_{k\ell}]$ is an $N\times N$ symmetric weighted adjacency matrix with $c_{k\ell}\geq 0$ if $\ell\in\cN_k$ and $c_{k\ell}=0$ otherwise. The smaller $S(\cw)$ is, the smoother the signal $\cw$ on the graph is. Since $L_c$ is symmetric positive semi-definite, it can be decomposed as $L_c=V\Lambda V^\top$ where $\Lambda=\diag\{\lambda_1,\ldots,\lambda_N\}$ with $\lambda_m$ the non-negative eigenvalues ordered as $0=\lambda_1\leq\lambda_2\leq\ldots\leq\lambda_N$ and $V=[v_1,\ldots,v_N]$ is the matrix of orthonormal eigenvectors. When the graph is connected, there is only one zero eigenvalue with corresponding eigenvector $v_1=\frac{1}{\sqrt{N}}\mathds{1}_N$~\cite{chung1997spectral}. Using the eigenvalue decomposition $\cL=(V\Lambda V^\top)\otimes I_L$, $S(\cw)$ can be written as:
\begin{equation}
S(\cw)=\cwb^\top(\Lambda\otimes I_L) \cwb=\sum_{m=1}^N\lambda_m\|\wb_m\|^2,
\end{equation}
where $\cwb=(V^\top\otimes I_L)\cw$ and $\wb_m=(v_m^\top\otimes I_L)\cw$. Given that $\lambda_m\geq 0$, the above expression shows that $\cw$ is considered to be smooth if $\|\wb_m\|^2$ corresponding to large $\lambda_m$ is negligible. Thus, for a smooth $\cw$, $S(\cw)$ will be equal to $\sum_{m=1}^p\lambda_m\|\wb_m\|^2$ with $p\ll N$. By choosing $\cU=U\otimes I_L$ where $U=[v_1,\ldots,v_p]$, the smooth signal $\cw$ will be in the range space of $\cU$ since it can be written as $\cw=\cU s$ with $s=\col\{\wb_1,\ldots,\wb_p\}$. Therefore, distributed inference problems under smoothness can be recast in the form~\eqref{eq: constrained optimization problem}.\qed

Before proceeding, note that, in some cases, one may find a family of matrices $\cA$ satisfying conditions~\eqref{eq: third condition in the proposition},~\eqref{eq: first condition on eigenvector}, and~\eqref{eq: second condition on eigenvector} under the sparsity constraints~\eqref{eq: condition 2 on A}. For example, in consensus optimization described in Remark~1 where $\cU=\frac{1}{\sqrt{N}}(\mathds{1}_N\otimes I_L)$, by ensuring that the underlying graph is strongly connected and by choosing any doubly-stochastic $A$ satisfying the sparsity constraints, the resulting matrix $\cA=A\otimes I_L$ will satisfy the required conditions. The same observation holds for coupled optimization problems described in Remark 2. Several policies for designing \emph{locally} doubly-stochastic matrices have been proposed in the literature~\cite{sayed2013diffusion,sayed2014adaptation,sayed2014adaptive}. For more general $\cU$, designing an $\cA$ satisfying conditions~\eqref{eq: condition 1 on A} and~\eqref{eq: condition 2 on A} can be written as the following feasibility problem:
\begin{equation}
\label{eq: feasibility problem A}
\begin{array}{cl}
\text{find}&\cA\\
\text{such that}&\cA\,\cU=\cU,~~\cU^*\cA=\cU^*,\\ 
&\rho(\cA-\cP_{\cu})< 1,\\ 
&[\cA]_{k\ell}=0, ~\text{if }\ell\notin\cN_k \text{ and }\ell\neq k,
\end{array}
\end{equation} 
which is challenging in general. Not all network topologies satisfying~\eqref{eq: condition 2 on A} guarantee the existence of an $\cA$ satisfying condition~\eqref{eq: condition 1 on A}. The higher the dimension of the signal subspace is, the greater the graph connectivity has to be. In the works~\cite{nassif2019distributed,barbarossa2009distributed}, it is assumed that the sparsity constraints~\eqref{eq: condition 2 on A} and the signal subspace lead to a feasible problem. That is, it is assumed that problem~\eqref{eq: feasibility problem A} admits at least one solution. As a remedy for the violation of such assumption, one may increase the network connectivity by increasing the transmit power of each node, i.e., adding more links~\cite{barbarossa2009distributed}. \cmag{In the accompanying Part~II~\cite{nassif2019adaptation},} we shall relax the feasibility assumption by considering the problem of finding an $\cA$ that minimizes the number of edges to be added to the original topology while satisfying the constraints~\eqref{eq: third condition in the proposition},~\eqref{eq: first condition on eigenvector}, and~\eqref{eq: second condition on eigenvector}. In this case, if the original topology leads to a feasible solution, then no links will be added. Otherwise, we assume that the designer is able to  add some links to make the problem feasible. 

In the following section, we consider that a feasible $\cA$ (topology) is computed by the designer and that its blocks $\{A_{k\ell}\}_{\ell\in\cN_k}$ are provided to agent $k$ in order to run algorithm~\eqref{eq: distributed solution}. We shall study the performance of~\eqref{eq: distributed solution} in the mean-square-error sense. \cblue{We shall consider the general complex case, in addition to the real case, since complex-valued combination matrix $\cA$ and data $\bx_{k,i}$ are important in several applications, as will be the case in the distributed beamforming application considered later in  Section~\ref{subsec: Distributed linearly-constrained-minimum-variance (LCMV) beamformer}.}

\section{Stability analysis}
\label{sec: Stochastic performance analysis}
\cmag{In this Part~I,} we shall establish mean-square-error stability by showing that, for each agent $k$, the error variance relative to $w^o_k$ enters a bounded region whose size is in the order of $\mu$, namely, $\limsup_{i\rightarrow\infty}\expec\|w^o_k-\bw_{k,i}\|^2=O(\mu)$. Then, \cmag{building on this result, we will assess in the accompanying Part~II~\cite{nassif2019adaptation}} the size of this mean-square error by deriving closed-form expression for the network mean-square-deviation (MSD) defined by~\cite{sayed2014adaptation}:
\begin{equation}
\label{eq: MSD definition}
\text{MSD}\triangleq \mu\lim_{\mu\rightarrow 0}\left(\limsup_{i\rightarrow\infty}\frac{1}{\mu}\expec\left(\frac{1}{N}\|\cw^o-\bcw_{i}\|^2\right)\right),
\end{equation}
where $\bcw_i\triangleq\col\{\bw_{k,i}\}_{k=1}^N$. In this way, we will be able to conclude that distributed strategies of the form~\eqref{eq: distributed solution} with small step-size are able to lead to reliable performance even in the presence of gradient noise. We will be able also to conclude that the iterates generated by the distributed implementation achieve the centralized steady-state performance. 

As explained in~\cite[Chap.~8]{sayed2014adaptation}, in the general case where $J_k(w_k)$ are not necessarily quadratic in the (complex) variable $w_k$, we need to track the evolution of both quantities $\bw_{k,i}$ and $(\bw_{k,i}^*)^\top$ in order to examine how the network is performing. Since $J_k(w_k)$ is real valued, the evolution of the complex conjugate iterates  $(\bw_{k,i}^*)^\top$ is given by:
\begin{equation}
\label{eq: distributed solution for conjugate iterates}
\left\lbrace
\begin{array}{rl}
\big(\bpsi_{k,i}^*)^\top=&\big(\bw_{k,i-1}^*\big)^\top-\mu\widehat{\nabla_{w^\top_k}J_k}(\bw_{k,i-1}),\\[0.65ex]
\big(\bw_{k,i}^*\big)^\top=&\sum\limits_{\ell\in\cN_k}(A_{k\ell}^*)^\top\big(\bpsi_{\ell,i}^*\big)^\top.
\end{array}
\right.
\end{equation}
Representations~\eqref{eq: distributed solution} and~\eqref{eq: distributed solution for conjugate iterates} can be grouped together into a single set of equations by introducing extended vectors of dimensions $2M_k\times 1$ as follows:
\begin{equation}
\label{eq: extended algorithm equations}
\left\lbrace
\begin{split}
\left[\begin{array}{c}
\bpsi_{k,i}\\
\big(\bpsi_{k,i}^*)^\top
\end{array}\right]&=\left[\begin{array}{c}
\bw_{k,i-1}\\
\big(\bw_{k,i-1}^*)^\top
\end{array}\right]-\mu\left[\begin{array}{c}
\widehat{\nabla_{w_k^*}J_k}(\bw_{k,i-1})\\[0.5ex]
\widehat{\nabla_{w^\top_k}J_k}(\bw_{k,i-1})
\end{array}\right]\\
\left[\begin{array}{c}
\bw_{k,i}\\
\big(\bw_{k,i}^*)^\top
\end{array}\right]&=\sum_{\ell\in\cN_k}\left[\begin{array}{cc}
A_{k\ell}&0\\
0&(A_{k\ell}^*)^\top
\end{array}\right]\left[\begin{array}{c}
\bpsi_{\ell,i}\\
\big(\bpsi_{\ell,i}^*)^\top
\end{array}\right].
\end{split}
\right.
\end{equation}

Therefore, when the data is complex, extended vectors and matrices need to be introduced in order to analyze the network evolution. The arguments and results presented in the analysis are applicable to both cases of real and complex data through the use of data-type variable $h$ defined in Table~\ref{table: variables table}. When the data is real-valued, the complex conjugate transposition should be replaced by the real transposition. Table~\ref{table: variables table} lists a couple of variables and symbols that will be used in the sequel for both real and complex data cases. The superscript ``$e$'' is used to refer to extended quantities. Although in the real data case no extended quantities should be introduced, we use the superscript ``$e$'' for both data cases for compactness of notation.  

\begin{table}
\caption{Definition of some variables used throughout the analysis. $\cI$ is a permutation matrix defined by~\eqref{eq: permutation matrix block}.}
\begin{center}
\begin{tabular}{c|c|c}
\hline
\hline
Variable&Real data case &Complex data case \\ \hline
Data-type variable $h$&1&2\\
Gradient vector&$\nabla_{w^\top_k}J_k(w_k)$&$\nabla_{w^*_k}J_k(w_k)$\\ 
Error vector $\bwt_{k,i}^e$&$\bwt_{k,i}$ from~\eqref{eq: error vector at node k}&$\left[\begin{array}{c}\bwt_{k,i}\\(\bwt_{k,i}^*)^\top \end{array}\right]$\\  [2.8ex]
Gradient noise $\bs_{k,i}^e(w)$&$\bs_{k,i}(w)$ from~\eqref{eq: gradient noise process}&$\left[\begin{array}{c}\bs_{k,i}(w)\\(\bs_{k,i}^*(w))^\top \end{array}\right]$\\ [2.8ex]
Bias vector $b_k^e$&$b_k$ from~\eqref{eq: bias vector}&$\left[\begin{array}{c}b_k\\(b_k^*)^\top \end{array}\right]$\\ [2.8ex]
$(k,\ell)$-th block of $\cA^e$ & $A_{k\ell}$&$\left[\begin{array}{cc}A_{k\ell}&0\\0&(A^*_{k\ell})^\top \end{array}\right]$\\[2.8ex]
Matrix $\cU^e$&$\cU$&$\cI^\top\left[\begin{array}{cc}\cU&0\\0&(\cU^*)^\top \end{array}\right]$\\ [2.8ex]
Matrix $\cJ^e_{\epsilon}$&$\cJ_{\epsilon}$ from~\eqref{eq: partitioning of cV-1}&$\left[
\begin{array}{cc}
\cJ_{\epsilon}&0\\
0&(\cJ^*_{\epsilon})^\top
\end{array}
\right]$\\ [2.8ex]
Matrix $\cV_{R,\epsilon}^e$&$\cV_{R,\epsilon}$ from~\eqref{eq: partitioning of cV-1}&\hspace{-1.5mm}$\cI^\top\left[
\begin{array}{cc}
\cV_{R,\epsilon}&0\\
0&(\cV_{R,\epsilon}^*)^\top
\end{array}
\right]$\\ [2.8ex]
Matrix $(\cV_{L,\epsilon}^e)^*$&$\cV_{L,\epsilon}^*$ from~\eqref{eq: partitioning of cV-1}&$\left[
\begin{array}{cc}
\cV_{L,\epsilon}^*&0\\
0&\cV_{L,\epsilon}^\top
\end{array}
\right]\cI$
\\
\hline
 \end{tabular}
\end{center}
\label{table: variables table}
\end{table}


\subsection{Modeling conditions}
We analyze~\eqref{eq: distributed solution} under conditions~\eqref{eq: condition 2 on A},~\eqref{eq: third condition in the proposition},~\eqref{eq: first condition on eigenvector}, and~\eqref{eq: second condition on eigenvector} on $\cA$, and the following assumptions on the risks $\{J_k(\cdot)\}$ and on the gradient noise processes $\{\bs_{k,i}(\cdot)\}$ defined as:
\begin{equation}
\label{eq: gradient noise process}
\bs_{k,i}(w)\triangleq\nabla_{w^*_k}J_k(w)-\widehat{\nabla_{w^*_k}J_k}(w).
\end{equation}
Before proceeding, we introduce the Hermitian Hessian \cblue{matrix} functions~\cite[Appendix~B]{sayed2014adaptation}:
\begin{align}
H_k(w_k)&\triangleq\nabla_{w_k}^2J_k(w_k),\qquad\qquad\qquad\qquad\qquad\qquad\qquad\qquad\qquad(hM_k\times hM_k)\\
&=\left\lbrace
\begin{array}{l}
\nabla_{w_k^\top}[\nabla_{w_k}J_k(w_k)],\qquad\qquad\qquad\qquad\qquad\qquad\qquad~\text{when the data is real }(M_k\times M_k)\\
\left[\begin{array}{c|c}
\nabla_{w_k^*}[\nabla_{w_k}J_k(w_k)]&(\nabla_{w_k^\top}[\nabla_{w_k}J_k(w_k)])^*\\[0.5ex]\hline
\nabla_{w_k^\top}[\nabla_{w_k}J_k(w_k)]&(\nabla_{w_k^*}[\nabla_{w_k}J_k(w_k)])^\top
\end{array}
\right]\qquad\quad\text{when the data is complex }(2M_k\times 2M_k)
\end{array}\right.\nonumber\\
\cH(\cw)&\triangleq\diag\left\{H_1(w_1),\ldots,H_N(w_N)\right\},\quad~(hM\times hM).\label{eq: block diagonal Hessian}
\end{align}
Note that, when $J^\text{glob}(\cw)=\sum_{k=1}^NJ_k(w_k)$, we have:
\begin{equation}
\label{eq: relation of the hessian in the real case}
\nabla_{\ccw}^2J^\text{glob}(\cw)=\left\lbrace
\begin{array}{ll}
\cH(\cw),&\text{when the data is real }\\
\cI\cH(\cw)\cI^\top,&\text{when the data is complex }
\end{array}
\right.
\end{equation} 
where $\cI$ is a permutation matrix given by:
\begin{equation}
\label{eq: first permutation matrix}
\cI\triangleq\left[
\begin{array}{cccccccc}
I_{M_1}&0&0&0&\ldots&0&0&0\\
0&0&I_{M_2}&0&\ldots&0&0&0\\
&&&&\ddots&&&\\
0&0&0&0&\ldots&0&I_{M_N}&0\\ \hline
0&I_{M_1}&0&0&\ldots&0&0&0\\
0&0&0&I_{M_2}&\ldots&0&0&0\\
&&&&\ddots&&&\\
0&0&0&0&\ldots&0&0&I_{M_N}
\end{array}
\right].
\end{equation} 
\cmag{This matrix consists of $2N\times 2N$ blocks with $(m,n)$-th block given by:
\begin{equation}
\label{eq: permutation matrix block}
[\cI]_{mn}\triangleq\left\lbrace
\begin{array}{ll}
I_{M_k},&\text{if } m=k, n=2(k-1)+1\\
I_{M_k},&\text{if } m=k+N, n=2k\\
0,&\text{otherwise}\\
\end{array}
\right.
\end{equation} 
for $m,n=1,\ldots,2N$ and $k=1,\ldots,N$. }

\begin{assumption}{\emph{(Conditions on aggregate and individual costs).}}
\label{assump: risks}
The individual costs $J_k(w_k)\in\mathbb{R}$ are assumed to be twice differentiable and convex such that:
\begin{equation}
\label{eq: convexity condition}
\frac{\nu_k}{h}I_{hM_k}\leq H_k(w_k)\leq\frac{\delta_{k}}{h}I_{hM_k},
\end{equation} 
where $\nu_{k}\geq 0$ for $k=1,\ldots,N$. It is further assumed that, for any $\cw$, $\cH(\cw)$ satisfies:
\begin{equation}
\label{assump: aggregate risk}
0<\frac{\nu}{h}I_{hP}\leq(\cU^e)^*\cH(\cw)\cU^e\leq\frac{\delta}{h}I_{hP},
\end{equation}
for some positive parameters $\nu\leq \delta$. The data-type variable $h$ and the matrix $\cU^e$ are defined in \emph{Table~\ref{table: variables table}}.
\end{assumption}
\noindent Condition~\eqref{assump: aggregate risk} ensures that problem~\eqref{eq: constrained optimization problem}, which can be rewritten as:
\begin{equation}
\label{eq: unconstrained optimization problem}
\cw^o=\cU s^o, \quad \text{with }s^o=\arg\min_{s}f(s)\triangleq J^{\text{glob}}(\cU s),
\end{equation}
has a unique minimizer $\cw^o$. This is due to the fact that the Hessian of $f(s)$, which is given by:
\begin{align}
&\nabla^2_{s}f(s)\notag\\
&=\left\lbrace
\begin{array}{l}
\cU^\top\left[\nabla^2_{\ccw}J^{\text{glob}}(\cw)\right]_{\ccw=\cU s}\cU,\qquad\qquad\qquad\qquad\qquad\qquad\quad\text{(real data case)}\\[0.5ex]
\diag\left\{\cU^*,\cU^\top\right\}\left[\nabla_{\ccw}^2J^\text{glob}(\cw)\right]_{\ccw=\cU s}\diag\left\{\cU,(\cU^*)^\top\right\},~\qquad\text{(complex data case)}
\end{array}
\right.\nonumber\\
&\hspace{-1mm}\overset{\eqref{eq: relation of the hessian in the real case}}=(\cU^e)^*\,\cH(\cU s)\,\cU^e\overset{\eqref{assump: aggregate risk}}{\geq}\frac{\nu}{h} I_{hP}>0,
\end{align}
is positive definite under condition~\eqref{assump: aggregate risk}.
\begin{assumption}{\emph{(Conditions on gradient noise).}}
\label{assump: gradient noise}
The gradient noise process defined in~\eqref{eq: gradient noise process} satisfies for any $\bw\in\bcF_{i-1}$ and for all $k,\ell=1,\ldots,N$:
\begin{align}
\expec[\bs_{k,i}(\bw)|\bcF_{i-1}]&=0,\label{eq: gradient noise mean condition}\\
\expec[\bs_{k,i}(\bw)\bs_{\ell,i}^*(\bw)|\bcF_{i-1}]&=0,\qquad k\neq \ell,\label{eq: uncorrelatedness of the gradient noises}\\
\expec[\bs_{k,i}(\bw)\bs_{\ell,i}^\top(\bw)|\bcF_{i-1}]&=0,\qquad k\neq \ell,\label{eq: circularity of the gradient noises}\\
\expec[\|\bs_{k,i}(\bw)\|^2|\bcF_{i-1}]&\leq(\beta_k/h)^2\|\bw\|^2+\sigma^2_{s,k},\label{eq: gradient noise mean square condition}
\end{align}
for some $\beta_k^2\geq 0$, $\sigma^2_{s,k}\geq 0$, and where $\bcF_{i-1}$ denotes the filtration generated by the random processes $\{\bw_{\ell,j}\}$ for all $\ell=1,\ldots,N$ and $j\leq i-1$. 
\end{assumption}
\noindent As explained in~\cite{sayed2013diffusion,sayed2014adaptation,sayed2014adaptive}, these conditions are satisfied by many objective functions of interest in learning and adaptation such as quadratic and logistic risks. 
Condition~\eqref{eq: gradient noise mean condition} essentially states that the gradient vector approximation should be unbiased conditioned on the past data, which is a reasonable condition to require. Condition~\eqref{eq: gradient noise mean square condition} states that the second-order moment of the gradient noise process should get smaller for better estimates, since it is bounded by the squared-norm of the iterate. Conditions~\eqref{eq: uncorrelatedness of the gradient noises} and~\eqref{eq: circularity of the gradient noises} state that the gradient noises across the agents are uncorrelated and second-order circular.

Without loss of generality, we shall introduce the following assumption on the matrix $\cU$\footnote{This assumption is not restrictive since for any full-column rank matrix $\cU' = [u_1',\ldots, u'_P]$ with $P\leq M$, we can generate by using, for example, the Gram-Schmidt process~\cite[pp.~15]{horn2013matrix}, a semi-unitary matrix $\cU = [u_1, \ldots, u_P]$ that spans the same $P$-dimensional subspace of $\mathbb{C}^M$ as $\cU'$, i.e., $\cR(\cU)=\cR(\cU')$.}. 
\begin{assumption}{\emph{(Condition on $\cU$).}}
\label{assump: matrix cU}
The full-column rank matrix $\cU$ is assumed to be semi-unitary, i.e., its column vectors are orthonormal and $\cU^*\cU=I_P$.
\end{assumption}

Before proceeding, we introduce an $N\times N$ block matrix $\cA^e$ whose $(k,\ell)$-th block is defined in Table~\ref{table: variables table}. This matrix will appear in our subsequent study. Observe that in the real data case, $\cA^e=\cA$, and that in the complex data case, $\cA^e$ can be seen as an extended version of the combination matrix $\cA$. The next statement exploits the eigen-structure of $\cA^e$ that will be useful for establishing the mean-square stability.
\begin{lemma}{\emph{(Jordan canonical decomposition).}}
\label{lemm: jordan decomposition}
Under Assumption~\ref{assump: matrix cU}, the $M\times M$ combination matrix $\cA$ satisfying conditions~\eqref{eq: first condition on eigenvector},~\eqref{eq: second condition on eigenvector}, and~\eqref{eq: third condition in the proposition} admits a Jordan canonical decomposition of the form:
\begin{equation}
\label{eq: eigendecomposition of cA}
\cA\triangleq\cV_{\epsilon}\Lambda_{\epsilon}\cV_{\epsilon},
\end{equation}
with:
\begin{equation}
\Lambda_{\epsilon}=\left[
\begin{array}{cc}
I_{P}&0\\
0&\cJ_{\epsilon}
\end{array}
\right],~\cV_{\epsilon}=\left[
\begin{array}{cc}
\cU&\cV_{R,\epsilon}
\end{array}
\right],~
\cV_{\epsilon}^{-1}=\left[
\begin{array}{c}
\cU^*\\
\cV_{L,\epsilon}^*
\end{array}
\right],\label{eq: partitioning of cV-1}
\end{equation}
where $\cJ_{\epsilon}$ is a Jordan matrix with the eigenvalues (which may be complex but have magnitude less than one) on the diagonal and $\epsilon>0$ on the  super-diagonal. It follows that the $hM\times hM$ matrix $\cA^e$ defined in Table~\ref{table: variables table} admits a  Jordan decomposition of the form:
\begin{equation}
\label{eq: jordan decomposition of A e}
\cA^e\triangleq\cV_{\epsilon}^e\Lambda_{\epsilon}^e(\cV_{\epsilon}^e)^{-1},
\end{equation}
with
\begin{equation}
\label{eq: partitioning of cVe-1}
\Lambda_{\epsilon}^e=\left[
\begin{array}{cc}
I_{hP}&0\\
0&\cJ^e_{\epsilon}
\end{array}
\right],~\cV^e_{\epsilon}=[\cU^e~\cV^e_{R,\epsilon}],~(\cV^e_{\epsilon})^{-1}=\left[
\begin{array}{c}
(\cU^e)^*\\
(\cV_{L,\epsilon}^e)^*
\end{array}
\right]
\end{equation}
where $\cU^e,\cJ^e_{\epsilon},\cV^e_{R,\epsilon}$, and $(\cV_{L,\epsilon}^e)^*$ are defined in \emph{Table~\ref{table: variables table}}. Since $(\cV_{\epsilon}^e)^{-1}\cV_{\epsilon}^e=I_{hM}$, the following relations hold:
\begin{equation}
\label{eq: equation from the eigendecomposition}
\begin{split}
&(\cU^e)^*\cU^e=I_{hP},\quad(\cV_{L,\epsilon}^e)^*\cV_{R,\epsilon}^e=I_{h(M-P)},\quad (\cU^e)^*\cV_{R,\epsilon}^e=0,\quad(\cV_{L,\epsilon}^e)^*\cU^e=0.
\end{split}
\end{equation}
\end{lemma}
\begin{proof}
\cmag{See Appendix~\ref{sec: jordan canonical decomposition}.}
\end{proof}

\subsection{Network error vector recursion}
Let $\bwt_{k,i}$ denote the error vector at node $k$:
\begin{equation}
\label{eq: error vector at node k}
\bwt_{k,i}\triangleq w^o_k-\bw_{k,i}.
\end{equation}
Consider first the complex data case. Using~\eqref{eq: gradient noise process} and the mean-value theorem~\cite[pp.~24]{polyak1987introduction},~\cite[Appendix D]{sayed2014adaptation}, we can express the stochastic gradient vectors appearing in~\eqref{eq: extended algorithm equations} as follows: 
\begin{equation}
\label{eq: mean-value theorem 1}
\left[\begin{array}{c}
\widehat{\nabla_{w_k^*}J_k}(\bw_{k,i-1})\\[0.5ex]
\widehat{\nabla_{w_k^\top}J_k}(\bw_{k,i-1})
\end{array}\right]=-\bH_{k,i-1}\bwt_{k,i-1}^e+b_k^e-\bs_{k,i}^e(\bw_{k,i-1})
\end{equation}
where:
\begin{equation}
\bH_{k,i-1}\triangleq\int_{0}^1{\nabla_{w_k}^2J_k}(w^o_k-t\bwt_{k,i-1})dt,\label{eq: Hki-1}
\end{equation}
and $\bwt_{k,i}^e$, $\bs_{k,i}^e(\bw_{k,i-1})$, and $b_k^e$ are defined in Table~\ref{table: variables table} with:
\begin{equation}
\label{eq: bias vector}
b_k\triangleq\nabla_{w^*_k}J_k(w^o_k).
\end{equation}
Subtracting $(w^o_k)^e=\col\{w^o_k,((w^o_k)^*)^\top\}$ from both sides of~\eqref{eq: extended algorithm equations} and by introducing the following extended vectors and matrices, which collect quantities from across the network:
\begin{align}
\bcwt_i^e&\triangleq\col\left\{\bwt^e_{1,i},\ldots,\bwt^e_{N,i}\right\},\label{eq: cwt}\\
\bcH_{i-1}&\triangleq\diag\left\{\bH_{1,i-1},\ldots,\bH_{N,i-1}\right\},\\
\bcB_{i-1}&\triangleq\cA^e(I_{hM}-\mu\bcH_{i-1}),\label{eq: matrix Bi-1 1}\\
\bs_i^e&\triangleq\col\left\{\bs^e_{1,i}(\bw_{1,i-1}),\ldots,\bs^e_{N,i}(\bw_{N,i-1})\right\},\\
b^e&\triangleq\col\left\{b_1^e,\ldots,b_N^e\right\},\label{eq: expression of b 1}
\end{align}
we can show that the network weight error vector $\bcwt^e_i$ in~\eqref{eq: cwt} \cblue{evolves} according to the following dynamics:
\begin{equation}
\label{eq: error recursion 1}
\fbox{$\bcwt_i^e=\bcB_{i-1}\bcwt^e_{i-1}-\mu\cA^e\bs_i^e+\mu\cA^e b^e$}
\end{equation}
where $\cA^e$ is defined in Table~\ref{table: variables table} and where we used~\eqref{eq: mean-value theorem 1} and the fact that 
\begin{equation}
\label{eq: relation between Awstar and wstar 2}
\sum_{\ell\in\cN_k}A_{k\ell}w^o_\ell=w^o_k,
\end{equation}
since $\cw^o$ is the solution of problem~\eqref{eq: constrained optimization problem}, and thus:
\begin{equation}
\label{eq: relation between Awstar and wstar}
\cA\cw^o=\cA\cP_\cu\cw^o\overset{\eqref{eq: first condition in the proposition}}{=}\cP_\cu\cw^o=\cw^o.
\end{equation}
For real data, the model can be simplified since we do not need to track the evolution of the complex conjugate $(\bw_{k,i}^*)^\top$. Although we use the notation ``$e$'' for the quantities in the above recursion, it is to be understood that the extended quantities $\{\bwt^e_{k,i},b_k^e,\bs_{k,i}^e,\cA^e\}$ should be replaced by the quantities $\{\bwt_{k,i},b_k,\bs_{k,i},\cA\}$ as in Table~\ref{table: variables table}.

The stability analysis of recursion~\eqref{eq: error recursion 1} is facilitated by transforming it to a convenient basis using the Jordan decomposition of $\cA^e$ in Lemma~\ref{lemm: jordan decomposition}. Multiplying both sides of~\eqref{eq: error recursion 1} from the left by $(\cV_{\epsilon}^e)^{-1}$ and introducing the transformed iterates and variables:
\begin{align}
(\cV_{\epsilon}^e)^{-1}\bcwt_i^e&=\left[
\begin{array}{c}
(\cU^e)^*\bcwt^e_i\\
(\cV_{L,\epsilon}^e)^*\bcwt_i^e
\end{array}
\right]\triangleq\left[\begin{array}{c}
\bcwb_{i}^e\\
\bcwc_{i}^e
\end{array}
\right],\label{eq: transformed version of wt}\\
\mu(\cV_{\epsilon}^e)^{-1}\cA^e\bs_i^e&=\left[
\begin{array}{c}
\mu\,(\cU^e)^*\cA^e\bs_i^e\\
\mu\,(\cV_{L,\epsilon}^e)^*\cA^e\bs_i^e
\end{array}
\right]\triangleq\left[\begin{array}{c}
\bsb_{i}^e\\
\bsc_{i}^e
\end{array}
\right],\label{eq: transformed version of the noise}\\
\mu(\cV_{\epsilon}^e)^{-1}\cA^e b^e&=\left[
\begin{array}{c}
\mu\,(\cU^e)^*\cA^e b^e\\
\mu\,(\cV_{L,\epsilon}^e)^*\cA^e b^e
\end{array}
\right]\triangleq\left[\begin{array}{c}
0\\
\widecheck{b}^e
\end{array}
\right],\label{eq: transformed version of the bias}
\end{align}
we obtain from Lemma~\ref{lemm: jordan decomposition}:
\begin{align}
\bcwb_{i}^e&=(I_{hP}-\bcD_{11,i-1})\bcwb_{i-1}^e-\bcD_{12,i-1}\bcwc_{i-1}^e-\bsb_{i}^e,\label{eq: error recursion for wbi}\\
\bcwc_{i}^e&=(\cJ_{\epsilon}^e-\bcD_{22,i-1})\bcwc_{i-1}^e-\bcD_{21,i-1}\bcwb_{i-1}^e-\bsc_{i}^e+\widecheck{b}^e,\label{eq: error recursion for wci}
\end{align}
where 
\begin{align}
\bcD_{11,i-1}&=\mu\,(\cU^e)^*\bcH_{i-1}\cU^e,\label{eq: D11}\\
\bcD_{12,i-1}&=\mu\,(\cU^e)^*\bcH_{i-1}\cV_{R,\epsilon}^e,\label{eq: D12}\\
\bcD_{21,i-1}&=\mu\,\cJ_{\epsilon}^e(\cV_{L,\epsilon}^e)^*\bcH_{i-1}\cU^e,\label{eq: D21}\\
\bcD_{22,i-1}&=\mu\,\cJ_{\epsilon}^e(\cV_{L,\epsilon}^e)^*\bcH_{i-1}\cV_{R,\epsilon}^e.\label{eq: D22}
\end{align}
Recursions~\eqref{eq: error recursion for wbi} and~\eqref{eq: error recursion for wci} can be written more compactly as:
\begin{align}
\left[\begin{array}{c}
\bcwb_{i}^e\\
\bcwc_{i}^e
\end{array}\right]=&
\left[\begin{array}{cc}
I_{hP}-\bcD_{11,i-1}&-\bcD_{12,i-1}\notag\\
-\bcD_{21,i-1}&\cJ_{\epsilon}^e-\bcD_{22,i-1}
\end{array}\right]\left[\begin{array}{c}
\bcwb_{i-1}^e\\
\bcwc_{i-1}^e
\end{array}\right]\\
&-\left[\begin{array}{c}
\bsb_{i}^e\\
\bsc_{i}^e
\end{array}\right]+\left[\begin{array}{c}
0\\
\widecheck{b}^e
\end{array}\right].\label{eq: error recursion transformed 1}
\end{align}
The zero entry in~\eqref{eq: transformed version of the bias} is due to the fact that 
\begin{equation}
\label{eq: equation on the bias}
(\cU^{e})^*\cA^e \,b^e\overset{\eqref{eq: jordan decomposition of A e},\eqref{eq: equation from the eigendecomposition}}{=}(\cU^{e})^* b^e=0,
\end{equation}
since the constrained optimization problem~\eqref{eq: constrained optimization problem} can be written alternatively as:
\begin{equation}
\label{eq: optimization problem 1}
\begin{split}
&\minimize_\ccw\quad J^\text{glob}(\cw)= \sum_{k=1}^NJ_k(w_k)\\
&\st\quad(I_{M}-\cP_{\cu})\cw=0.
\end{split}
\end{equation}
The Lagrangian associated with problem~\eqref{eq: optimization problem 1} is given by:
\begin{equation}
\cL(\cw;\gamma)=J^\text{glob}(\cw)+h\text{Re}\{\gamma^*(I_{M}-\cP_\cu)\cw\}
\end{equation}
where $\gamma$ is the $M\times 1$ vector of Lagrange multipliers. From the optimality conditions, we obtain the following condition on $\cw^o$:
\begin{equation}
\label{eq: bias vector first component}
\left\lbrace
\begin{array}{l}
\underbrace{\nabla_{\ccw^\top}J^{\text{glob}}(\cw^o)}_{b^e}+(I_{M}-\cP_\cu)\gamma=0,\qquad\qquad\qquad\qquad\qquad\qquad\qquad\qquad\text{(when the data is real)}\\
\underbrace{\left[\begin{array}{c}
\nabla_{\ccw^*}J^{\text{glob}}(\cw^o)\\
\nabla_{\ccw^\top}J^{\text{glob}}(\cw^o)
\end{array}
\right]}_{\cI b^e}+\left[\begin{array}{cc}
I_M-\cP_{\cu}&0\\
0&(I_M-\cP_{\cu})^\top
\end{array}
\right]\left[\begin{array}{c}
\gamma\\
(\gamma^*)^\top
\end{array}
\right]=0,\qquad\text{(when the data is complex)}
\end{array}
\right.
\end{equation}
where we used the fact that $\sum_{k=1}^NJ_k(w_k)$ is real valued and where $b^e$ and $\cI$ are given by~\eqref{eq: expression of b 1} and~\eqref{eq: permutation matrix block}, respectively. In the real data case, by multiplying both sides of the previous relation by $(\cU^e)^*=\cU^\top$, we obtain $(\cU^e)^* b^e=0$. For complex data, by multiplying both sides of the previous equation by $(\cU^e)^*\cI^\top$ with $\cU^e$ defined in Table~\ref{table: variables table}, we obtain $(\cU^e)^* b^e=0$. Now, considering both real and complex data cases, we arrive at~\eqref{eq: equation on the bias}.

\noindent\textbf{Remark 5.} Regarding algorithm~\eqref{eq: distributed solution affine constrained optimization problem}, it can be verified that the weight error vector $\bcwt_i^e$ will end up evolving according to recursion~\eqref{eq: error recursion 1}. The constant driving terms $\{d_{\cd,k}\}$ will disappear when subtracting $w^o_k$ from both sides of~\eqref{eq: distributed solution affine constrained optimization problem} since $w^o_k$ satisfies the following relation:
\begin{equation}
w^o_k=\sum_{\ell\in\cN_k}A_{k\ell}w^o_{\ell}+d_{\cd,k},
\end{equation} 
where we used the fact that the optimal solution $\cw^o$ in~\eqref{eq: affine constrained optimization problem} verifies:
\begin{align}
\cw^o=\cP_\cu\cw^o+d_{\cd}&\overset{\eqref{eq: first condition in the proposition}}{=}\cA(\cP_\cu\cw^o+d_{\cd})+(I_M-\cA)d_{\cd}\notag\\
&~=\cA\cw^o+(I_M-\cA)d_{\cd}.
\end{align}
By rewriting the constraint in~\eqref{eq: affine constrained optimization problem} as $(I_M-\cP_\cu)\cw=d_{\cd}$ and repeating similar arguments as~\eqref{eq: optimization problem 1}--\eqref{eq: bias vector first component}, we can show that $(\cU^e)^* b^e=0$. Therefore, the transformed iterates $\bcwb_{i}^e$ and $\bcwc_{i}^e$ in~\eqref{eq: transformed version of wt} will continue to evolve according to recursions~\eqref{eq: error recursion for wbi} and~\eqref{eq: error recursion for wci}. \qed

In the following, we shall establish the mean-square-error stability of algorithm~\eqref{eq: distributed solution}. \cmag{In the accompanying Part~II, we will \cblue{derive a}} closed-form expression for the network MSD defined by~\eqref{eq: MSD definition}. The derivation is demanding. \cmag{However, the arguments are along the lines developed in~\cite[Chaps. 9--11]{sayed2014adaptation} for standard diffusion~\eqref{eq: diffusion strategy} with proper adjustments to handle possibly complex valued {\em block} matrices $\{A_{k\ell}\}$ satisfying conditions~\eqref{eq: condition 1 on A} and~\eqref{eq: condition 2 on A} and {the} subspace constraints. }

\subsection{Mean-square-error stability}
\label{subsec: Mean-square-error stability}
\begin{theorem}{\emph{(Network mean-square-error stability).}}
\label{theo: Network mean-square-error stability}
Consider a network of $N$ agents running the distributed  strategy~\eqref{eq: distributed solution} with a matrix $\cA$ satisfying conditions~\eqref{eq: first condition on eigenvector},~\eqref{eq: second condition on eigenvector}, and~\eqref{eq: third condition in the proposition} and $\cU$ satisfying Assumption~\ref{assump: matrix cU}. Assume the individual costs, $J_k(w_k)$, satisfy the conditions in Assumption~\ref{assump: risks}. Assume further that the first and second-order moments of the gradient noise process satisfy the conditions in Assumption~\ref{assump: gradient noise}. Then, the network is mean-square-error stable for sufficiently small step-sizes, namely, it holds that:
\begin{equation}
\label{eq: mean-square error convergence result}
\limsup_{i\rightarrow\infty}\expec\|w^o_k-\bw_{k,i}\|^2=O(\mu),\quad k=1,\ldots,N,
\end{equation} 
for small enough $\mu$.
\end{theorem}
\begin{proof}
See Appendix~\ref{app: mean-square error stability}.
\end{proof}

\section{Distributed linearly constrained minimum variance (LCMV) beamformer}
\label{subsec: Distributed linearly-constrained-minimum-variance (LCMV) beamformer}

Consider a uniform linear array (ULA) of $N=14$ antennas, as shown in Fig.~\ref{fig: beamformer}. A desired narrow-band signal $\bs_0(i)\in\mathbb{C}$ from far field impinges on the array from known direction of arrival (DOA) $\theta_0=30^\circ$ along with two uncorrelated interfering signals $\{\bs_1(i),\bs_2(i)\}\in\mathbb{C}$ from DOAs $\{\theta_1=-60^{\circ},\theta_2=60^{\circ}\}$, respectively. We assume that the DOA of $\bs_3(i)$ is roughly known. The received signal at the array is therefore modeled as:
\begin{equation}
\bx_i=a(\theta_0)\bs_0(i)+\underbrace{\sum_{n=1}^2a(\theta_n)\bs_n(i)+\bv_i}_{\bz_i}
\end{equation}
where $\bx_i=\col\{\bx_1(i),\ldots,\bx_N(i)\}$ is an $N\times 1$ vector that collects the received signals at the antenna elements, $\{a(\theta_n)\}_{n=0}^2$ are $N\times 1$ array manifold vectors (steering vectors) for the desired and interference signals, and $\bv_i=\col\{\bv_1(i),\ldots,\bv_N(i)\}$ is the additive noise vector at time $i$. With the first element as the reference point, the $N\times 1$ array manifold vector $a(\theta_n)$ is given by
$a(\theta_n)=\col\left\{1,e^{-j\tau_n},e^{-j2\tau_n},\ldots,e^{-j(N-1)\tau_{n}}\right\}$~\cite{trees2002optimum},
with $\tau_n=\frac{2\pi d}{\lambda}\sin(\theta_n)$ where $d$ denotes the spacing between two adjacent antenna elements, and $\lambda$ denotes the wavelength of the carrier signal. The antennas are assumed spaced half a wavelength apart, i.e., $d=\lambda/2$. 

Beamforming problems generally deal with the design of a weight vector $h=\col\{h_1,\ldots,h_N\}\in\mathbb{C}^{N\times 1}$ in order to recover the desired signal $\bs_0(i)$ from the received data $\bx_i$~\cite{trees2002optimum,frost1972algorithm}. The narrowband beamformer output can be expressed as $\by(i)=h^*\bx_i$. Among many possible criteria, we  use the linearly-constrained-minimum-variance (LCMV) design, namely,
\begin{equation}
\label{eq: lcmv beamforming problem}
\begin{split}
h^o=&\,\arg\min_h~~\expec\,|h^*\bx_i|^2=h^*R_x h\\
&\,\st ~D^*h=b,
\end{split}
\end{equation}
where $D$ is an $N\times P$ matrix and $b$ is a $P\times 1$ vector, in order to suppress the influence of the perturbation $\bz_i$ on the output $\by(i)$ while preserving the signal component. Since the DOAs of $\bs_0(i)$ is known and the DOA of $\bs_2(i)$ is roughly known, matrix $D$ can be chosen as $D=[a(30^\circ)~a(58.5^\circ)~a(61.5^\circ)]$, and the vector $b$ as $b=\col\{1,0.01,0.01\}$. In this way, we set unit response to the direction of the desired signal so that $\bs_0(i)$ passes through the array without distortion.

In a distributed setting, the objective of agent (antenna element) $k$ is to estimate $h^o_k$, the $k$-th component of $h^o$ in~\eqref{eq: lcmv beamforming problem}. Neighboring agents are allowed to exchange their observations $\bx_\ell(i)$.  To each agent $k$, we associate a neighborhood set $\cN_k$, an $M_k\times 1$ parameter vector $w_k$, and an $M_k\times 1$ regression vector $\bu_{k,i}$, defined in Table~\ref{table: beamformer table} depending on the node location on the array.  Observe that the parameter $\nu$ controls the network topology. For example, $\nu=N-1$ corresponds to a fully connected network setting. We associate with each agent $k$ a cost $J_k(w_k)\triangleq w^*_k\expec[\bu_{k,i}\bu_{k,i}^*]w_k.$
Instead of solving~\eqref{eq: lcmv beamforming problem}, we propose to solve:
\begin{equation}
\label{eq: distributed lcmv beamforming problem}
\begin{split}
\cw^o=&\,\arg\min_\ccw~~J^{\text{glob}}(\cw)=\sum_{k=1}^NJ_k(w_k)\\
&\,\st ~\cD^*\cw=d,
\end{split}
\end{equation}
where the equality constraint $\cD^*\cw=d$ merges the equality constraint in~\eqref{eq: lcmv beamforming problem} and the equality constraints that need to be imposed on the parameter vectors at neighboring nodes in order to achieve equality between common entries (see Table~\ref{table: beamformer table}). Let $E$ denote the binary connection matrix with $[E_{k\ell}]=1$ if $\ell\in\cN_k$, and $0$ otherwise. Under the consensus constraints, it can be shown that:
\begin{equation}
J^{\text{glob}}(\cw)=\sum_{k=1}^NJ_k(w_k)=h^*(F\circ R_x) h
\end{equation}
where $F$ is an $N\times N$ matrix with $[F]_{k\ell}=\frac{[E^2]_{k\ell}}{\sqrt{|\cN_k||\cN_{\ell}|}}$ and $\circ$ is the element-wise product. Therefore, collecting observations from neighboring nodes allows partial covariance matrix computation, which will be used in optimization. For the partial covariance $F\circ R_x$ to converge to the true covariance $R_x$ in~\eqref{eq: lcmv beamforming problem}, we need to set $\nu=N-1$ in order to have $F=\mathds{1}_N\mathds{1}_N^\top$. Note that, two main classes of distributed beamforming appear in the literature~\cite{koutrouvelis2018lowcost}. In the first class, which is considered here, the covariance matrix is approximated to form distributed implementations~\cite{connor2014diffusion,zeng2014distributed,connor2016distributed,koutrouvelis2018lowcost} leading to sub-optimal beamformers. In the second class, the proposed beamformers obtain statistical optimality but do so at the expense of restricting the topology of the underlying network~\cite{bertrand2013distributed}. Different from~\cite{connor2014diffusion}, the current distributed solution preserves convexity and is scalable since nodes exchange and compute $M_\ell\times 1$ sub-vectors  $\{w_\ell\}$ with $M_\ell=|\cN_\ell|<N$ instead of $N\times 1$ vectors. 
\begin{table*}
\caption{Distributed beamforming settings for uniform linear arrays of $N$ antennas ($1\leq \nu\leq N-1$).}
\begin{center}
\begin{tabular}{|c|c|c|}
\hline
\cellcolor{gray!60}Neighboring set $\cN_k$& \cellcolor{gray!60}Parameter vector $w_k$& \cellcolor{gray!60}Regressor $\bu_{k,i}$ \\ \hline\hline
$\{\max\{1,k-\nu\},\ldots,\min\{k+\nu,N\}\}$&$\col\{h_m\}_{m=\max\{1,k-\nu\}}^{\min\{k+\nu,N\}}$&$\col\{|\cN_m|^{-\frac{1}{2}}\bx_m(i)\}_{m=\max\{1,k-\nu\}}^{\min\{k+\nu,N\}}$\\ \hline
 \end{tabular}
\end{center}
\label{table: beamformer table}
\vspace{-3mm}
\end{table*}
\begin{figure*}
\begin{center}
\includegraphics[scale=0.25]{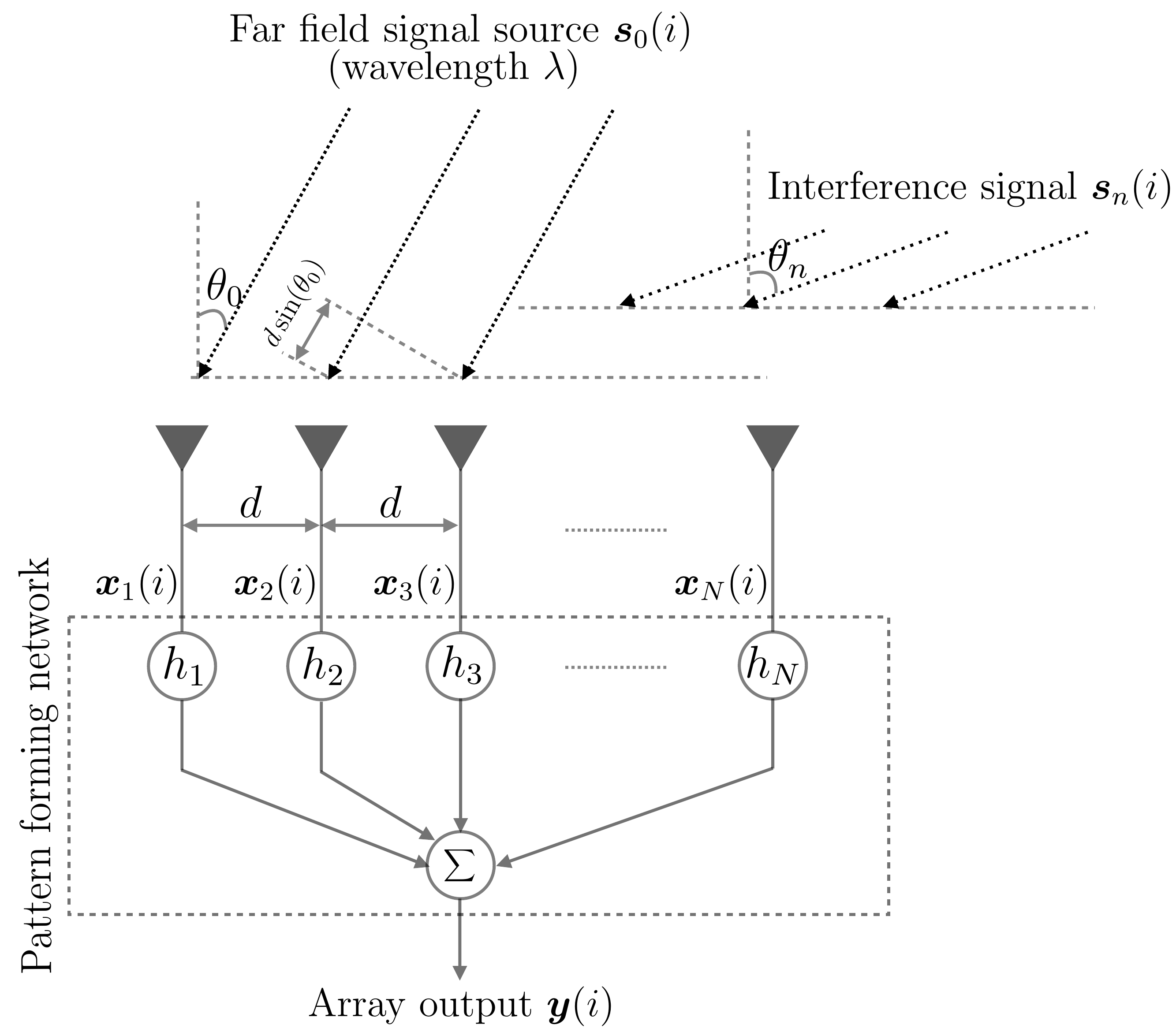}\qquad
\includegraphics[scale=0.4]{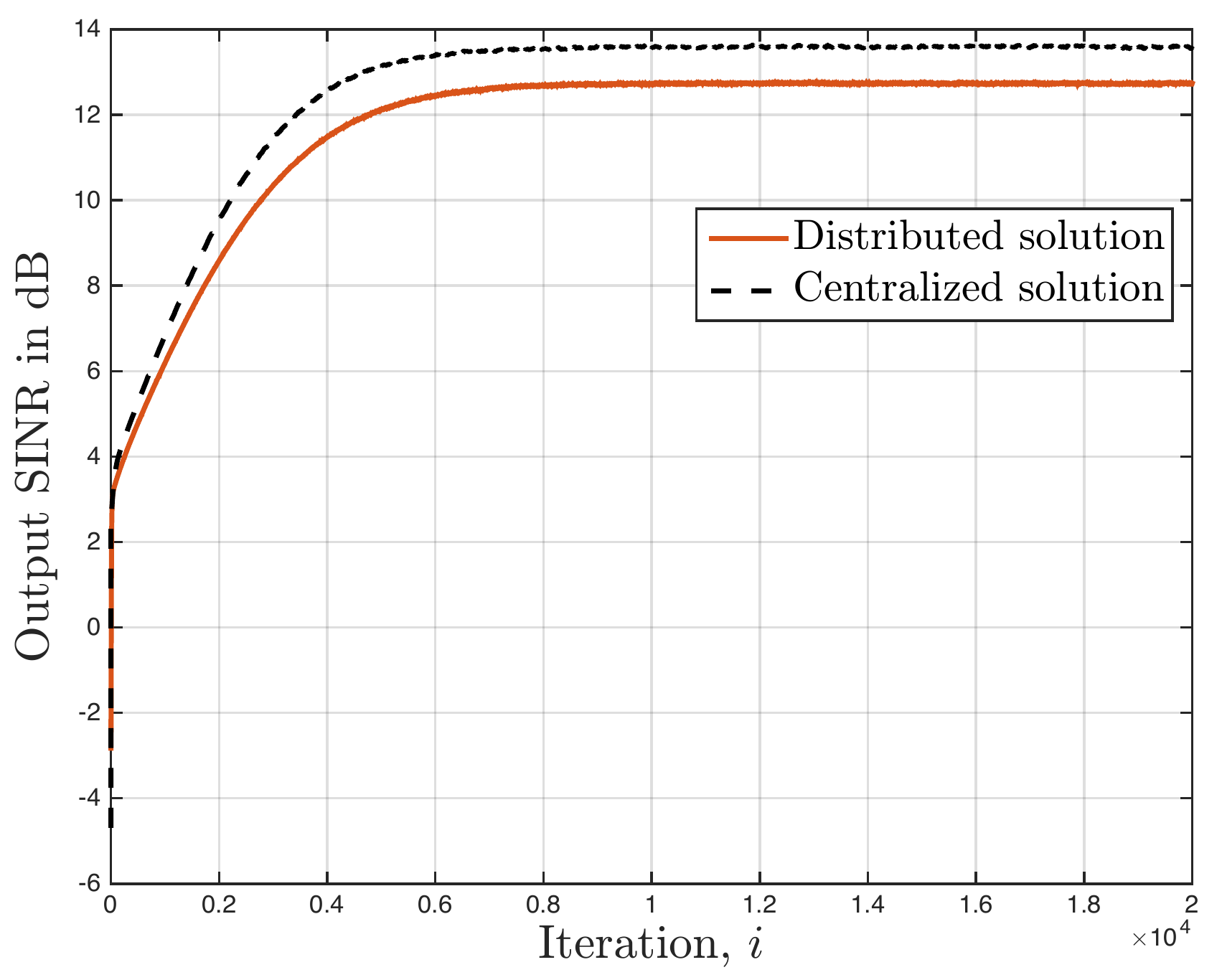}
\vspace{-2mm}
\caption{\emph{(Left)} Uniform linear array of $N$ antennas. \emph{(Right)} Comparison of output SINR.}
\label{fig: beamformer}
\end{center}
\vspace{-4mm}
\end{figure*}

Algorithm~\eqref{eq: distributed solution} can be applied to solve~\eqref{eq: distributed lcmv beamforming problem}. The signals $\{\bs_n(i)\}_{n=0}^2$ are i.i.d. zero-mean complex Gaussian random variables with variance $\sigma^2_{s,n}=1,\forall n$. The additive noise $\bv_i$ is zero-mean complex Gaussian with covariance $\expec\bv_i\bv_i^*=\sigma_v^2I_N$ ($\sigma_v=0.7$). We set $\nu=4$. The complex combination matrix  $\cA$  is set as the solution of the \cmag{feasibility problem~\eqref{eq: feasibility problem A} with the constraint $\rho(\cA-\cP_\cu)<1$ replaced by $\rho(\cA-\cP_\cu)\leq 1-\epsilon$ ($\epsilon=0.01$) and the constraint $\cA=\cA^*$ added\footnote{These changes make the problem convex--see~\cite[{Sec.~3}]{nassif2019adaptation} for further details.}}. The resulting problem is solved via CVX package~\cite{grant2014cvx}. Note that the distributed implementation is feasible \cmag{in this example}. We set $\mu=0.005$. 
 The output signal-to-interference-plus-noise ratio (SINR) given by $\expec\left[\frac{\sigma^2_{s,0}|\boldh_i^*a(\theta_0)|}{\boldh_i^*R_{z}\boldh_i}\right]$ with $R_z=\sum_{n=1}^2\sigma^2_{s,n}a(\theta_n)a^*(\theta_n)+\sigma^2_v I_N$ is illustrated in Fig.~\ref{fig: beamformer} (right).  The dashed black curve is the beampattern obtained by the centralized, also known as the constrained LMS~\cite{frost1972algorithm}, algorithm ($\mu=0.001$). The results are averaged over $1000$ Monte-Carlo runs. We observe that the distributed solution performs well compared to the centralized implementation.
 \vspace{-1mm}
 \section{Conclusion}
 \label{sec: conclusion}
In this work, we considered inference problems over networks where agents have individual parameter vectors to estimate subject to subspace constraints that require the parameters across the network to lie in low-dimensional subspaces. Based on the gradient projection algorithm, we proposed an iterative and distributed implementation of the projection step, which runs in parallel with the stochastic gradient descent update. We showed that, for small step-size parameter, the network is able to approach the minimizer of the constrained problem to arbitrarily good accuracy levels. 
\begin{appendices}
\vspace{-1mm}
\section{Proof of Lemma~\ref{lemm: matrix equation proposition}}
\label{app: proof of lemma 1}
First we prove sufficiency by proving that if $\cP_{\cu}$ is a projection matrix and $\cA$ satisfies conditions~\eqref{eq: first condition in the proposition},~\eqref{eq: second condition in the proposition}, and~\eqref{eq: third condition in the proposition}, then the matrix equation~\eqref{eq: condition 1 on A} holds. 
If $\cA$ satisfies~\eqref{eq: first condition in the proposition} and~\eqref{eq: second condition in the proposition}, then:
\begin{eqnarray}
\cA^i-\cP_{\cu}&\overset{\eqref{eq: first condition in the proposition}}{=}&\cA^i-\cA^i\cP_{\cu}\notag\\
&=&\cA^i(I-\cP_{\cu})\notag\\
&=&\cA^i(I-\cP_{\cu})^i\notag\\
&\overset{\eqref{eq: first condition in the proposition},\eqref{eq: second condition in the proposition}}{=}&(\cA(I-\cP_{\cu}))^i\notag\\
&\overset{\eqref{eq: first condition in the proposition}}{=}&(\cA-\cP_{\cu})^i
\end{eqnarray}
where we used the fact that $(I-\cP_{\cu})=(I-\cP_{\cu})^i$ since $(I-\cP_{\cu})$ is a projector. Applying condition~\eqref{eq: third condition in the proposition} and using the fact that for any matrix $B$, $\lim_{i\rightarrow\infty}B^i=0$ if and only if $\rho(B) < 1$, we obtain the desired convergence~\eqref{eq: condition 1 on A}. 

To prove necessity, we shall prove that every time we have~\eqref{eq: condition 1 on A}, we will have $\cP_{\cu}$  a projection matrix and conditions~\eqref{eq: first condition in the proposition},~\eqref{eq: second condition in the proposition}, and~\eqref{eq: third condition in the proposition} on $\cA$ satisfied. 
We use the fact that $\lim_{i\rightarrow\infty}\cA^i$ exists if, and only if, there is a non singular matrix $\cV$ such that~\cite{oldenburger1940}:
\begin{equation}
\label{eq: jordan decomposition of A}
\cA=\cV\left[
\begin{array}{cc}
I_K&0\\
0&\cJ
\end{array}
\right]\cV^{-1},
\end{equation}
where the spectral radius of $\cJ$ is less than one. Let $v_1,\ldots,v_{M}$ be the columns of $\cV$ and $y_1^*,\ldots,y_{M}^*$ be the rows of $\cV^{-1}$. Then, we have:
\begin{align}
\lim_{i\rightarrow\infty}\cA^i&=\lim_{i\rightarrow\infty}\cV\left[
\begin{array}{cc}
I_K&0\\
0&\cJ^i
\end{array}
\right]\cV^{-1}\notag\\
&=\cV\left[
\begin{array}{cc}
I_K&0\\
0&0
\end{array}
\right]\cV^{-1}=\sum_{m=1}^Kv_{m}y_m^*.\label{eq: jordan decomposition}
\end{align}
From the left hand-side of~\eqref{eq: condition 1 on A} and~\eqref{eq: jordan decomposition}, we obtain:
\begin{equation}
\label{eq: P R (U)}
\cP_{\cu}=\lim_{i\rightarrow\infty}\cA^i=\sum_{m=1}^Kv_{m}y_m^*\cblue{=\cV\left[
\begin{array}{cc}
I_K&0\\
0&0
\end{array}
\right]\cV^{-1}.}
\end{equation}
Observe from~\eqref{eq: jordan decomposition of A} that one is an eigenvalue of $\cA$ with multiplicity $K$ and $\{v_m,y_m\}_{m=1}^K$ are the associated right and left eigenvectors. Thus, from~\eqref{eq: P R (U)}, we obtain:
\begin{align}
\cA \cP_{\cu}=\cA\sum_{m=1}^Kv_{m}y_m^*=\sum_{m=1}^K v_{m}y_m^*=\cP_{\cu}\\
\cP_{\cu}\cA=\sum_{m=1}^Kv_{m}y_m^*\cA=\sum_{m=1}^K v_{m}y_m^*=\cP_{\cu}
\end{align}
and equations~\eqref{eq: first condition in the proposition} and~\eqref{eq: second condition in the proposition} hold. Moreover, from~\eqref{eq: jordan decomposition of A} and~\eqref{eq: P R (U)}, we obtain:
\begin{equation}
\rho\left(\cA-\cP_{\cu}\right)=\rho\left(\cV\left[
\begin{array}{cc}
0&0\\
0&\cJ
\end{array}
\right]\cV^{-1}\right)=\rho(\cJ)<1,
\end{equation}
which is condition~\eqref{eq: third condition in the proposition}. Finally, from~\eqref{eq: P R (U)}, we have:
\begin{equation}
\label{eq: P R (U) 3}
\cP_{\cu}^2=\cV\left[
\begin{array}{cc}
I_K&0\\
0&0
\end{array}
\right]\cV^{-1}=\cP_{\cu}.
\end{equation}
Thus, $\cP_\cu$ is a projector, which completes the necessity proof.

Since each $v_my_m^*$ is a rank-one matrix and their sum $\sum_{m=1}^{M}v_my_m^*=\cV\cV^{-1}=I$ has rank $M$, the matrix $\sum_{m=1}^Kv_my_m^*$ must have rank $K$. Since the rank of $\cP_{\cu}$ is equal to $P$, we obtain from~\eqref{eq: P R (U)} $K=P$. Thus, the matrix $\cA$ has $P=K$ eigenvalues at one and all other eigenvalues are strictly less than one.

\section{Driving term in algorithm~\eqref{eq: distributed solution affine constrained optimization problem}}
\label{app: affine constraints driving term}
Let $\cw_0$ denote an $M\times 1$ vector distributed across the network. In order to justify the choice of $(I-\cA) d_{\cd}$ in~\eqref{eq: distributed solution affine constrained optimization problem}, we consider the problem of finding the projection $\cw^o=\cP_{\cu}\cw_0+d_{\cd}$ in a distributed and iterative manner through a linear iteration of the form:
\begin{equation}
\label{eq: recursion 2 for simple case}
\cw_i=\cA\cw_{i-1}+\cB d_{\cd},
\end{equation}
where $\cA$ satisfies~\eqref{eq: condition 1 on A},~\eqref{eq: condition 2 on A} and $\cB$ is a properly chosen matrix  ensuring convergence toward $\cw^o$. Starting from $\cw_0$ and iterating the above recursion, we obtain:
\begin{equation}
\label{eq: recursion 0 for simple case}
\cw_i=\cA^i\cw_{0}+\sum_{j=0}^{i-1}\cA^j\cB d_{\cd}.
\end{equation}
If we let $i\rightarrow\infty$ on both sides of~\eqref{eq: recursion 0 for simple case}, we find:
\begin{equation}
\label{eq: recursion 1 for simple case}
\cw_{\infty}=\lim_{i\rightarrow\infty}\cw_i=\underbrace{\lim_{i\rightarrow\infty}\cA^i}_{\cP_{\cu}}\cw_{0}+\sum_{j=0}^{\infty}\cA^j\cB d_{\cd}.
\end{equation}
For $\cw_{\infty}$ to be equal to $\cw^o$, $\cB$ in~\eqref{eq: recursion 2 for simple case} must be chosen such that $\sum_{j=0}^{\infty}\cA^j\cB = I$. In the following, we show that $\cB=I-\cA+\cP_{\cu}$ ensures convergence. From the Jordan canonical form of $\cA$ introduced in~\eqref{eq: partitioning of cV-1}, we have:
\begin{align}
\cA^j&=\cV_{\epsilon}\left[\begin{array}{cc}
I_P&0\\
0&\cJ_{\epsilon}^i
\end{array}
\right]\cV_{\epsilon}^{-1},\\
 I-\cA+\cP_{\cu}&=\cV_{\epsilon}\left[\begin{array}{cc}
I_P&0\\
0&I-\cJ_{\epsilon}
\end{array}
\right]\cV_{\epsilon}^{-1}.
\end{align} 
If we multiply both terms and compute the infinite sum in~\eqref{eq: recursion 1 for simple case}, we obtain:
\begin{equation}
\sum_{j=0}^\infty \cA^j(I-\cA+\cP_{\cu})=\cV_{\epsilon}\left[\begin{array}{cc}
I_P&0\\
0&\sum_{j=0}^\infty\cJ_{\epsilon}^i(I-\cJ_{\epsilon})
\end{array}
\right]\cV_{\epsilon}^{-1}=I,
\end{equation} 
where we used the fact that $\sum_{j=0}^\infty\cJ_{\epsilon}^i=I-\cJ_{\epsilon}$ since $\rho(\cJ_{\epsilon})<1$. 

Now, since $\cP_{\cu}=\cP_{\cd}$ and $\cP_{\cd}d_{\cd}=0$, we obtain $\cB d_{\cd}=(I-\cA) d_{\cd}$, which justifies the choice in~\eqref{eq: distributed solution affine constrained optimization problem}.

\section{Proof of Lemma~\ref{lemm: jordan decomposition}}
\label{sec: jordan canonical decomposition}
We start by noting that the $M\times M$ matrix $\cA$ satisfying conditions~\eqref{eq: first condition on eigenvector},~\eqref{eq: second condition on eigenvector}, and~\eqref{eq: third condition in the proposition} admits a Jordan canonical decomposition of the form:
\begin{align}
\label{eq: jordan decomposition of cA}
\cA=\cV\Lambda\cV^{-1},\quad
\Lambda=\left[
\begin{array}{cc}
I_P&0\\
0&\cJ
\end{array}
\right],\quad\cV=\left[
\begin{array}{cc}
\cU&\cV_R
\end{array}
\right],\quad
\cV^{-1}=\left[
\begin{array}{c}
\cU^*\\
\cV_L^*
\end{array}
\right]
\end{align}
where the matrix $\cJ$ consists of Jordan blocks, with each one of them having the form (say for a Jordan block of size $3\times 3$):
\begin{equation}
\label{eq: jordan block}
\left[
\begin{array}{ccc}
\lambda&1&0\\
0&\lambda&1\\
0&0&\lambda
\end{array}
\right],
\end{equation}
where the eigenvalue $\lambda$ may be complex but has magnitude less than one. Let $\cE=\diag\{I_P,\epsilon,\epsilon^2,\ldots,\epsilon^{M-P}\}$ with $\epsilon>0$ any small positive number independent of $\mu$. The matrix $\cA$ in~\eqref{eq: jordan decomposition of cA} can be written alternatively as:
\begin{equation}
\label{eq: jordan canonical form of A}
\begin{split}
\cA=\underbrace{\cV\cE}_{\cV_{\epsilon}}\underbrace{\cE^{-1}\Lambda\cE}_{\Lambda_{\epsilon}}\underbrace{\cE^{-1}\cV^{-1}}_{\cV_{\epsilon}^{-1}}=\cV_{\epsilon}\Lambda_{\epsilon}\cV_{\epsilon}^{-1}
\end{split}
\end{equation}
where
\begin{align}
\Lambda_{\epsilon}=\left[
\begin{array}{cc}
I_P&0\\
0&\cJ_{\epsilon}
\end{array}
\right],~\quad\cV_{\epsilon}=\left[
\begin{array}{cc}
\cU&\cV_{R,\epsilon}
\end{array}
\right],~\quad\cV_{\epsilon}^{-1}=\left[
\begin{array}{c}
\cU^*\\
\cV_{L,\epsilon}^*
\end{array}
\right],\label{eq: partitioning of cV-1 2}
\end{align}
and where the matrix $\cJ_{\epsilon}$ consists of Jordan blocks, \cmag{with each one of them having a form similar as~\eqref{eq: jordan block} 
with $\epsilon>0$ appearing on the upper diagonal instead of $1$}, and where the eigenvalue $\lambda$ may be complex but has magnitude less than one.  Obviously, since $\cV_{\epsilon}^{-1}\cV_{\epsilon}=I_{M}$, it holds that:
\begin{align}
\cU^{*}\cV_{R,\epsilon}=0,\quad
\cV_{L,\epsilon}^{*}\cU=0,\quad
\cV_{L,\epsilon}^{*}\cV_{R,\epsilon}=I_{M-P},
\end{align}
and where $\cU^{*}\cU=I_{P}$ from Assumption~\ref{assump: matrix cU}. 

Now, let us consider the extended version of the matrix $\cA$, namely,  $\cA^e$, which is an $N\times N$ block matrix whose $(k,\ell)$-th block \cmag{is defined in Table~\ref{table: variables table}.}
 In the real data case, we have $\cA^e=\cA$. In the complex data case, it can be verified that $\cA^e$ is similar to the $2\times 2$ block diagonal matrix:
\begin{equation}
\label{eq: block diagonal form of A e}
\cA^d=\left[\begin{array}{cc}
\cA&0\\
0&(\cA^*)^\top
\end{array}\right]
\end{equation} 
according to:
\begin{equation}
\label{eq: similarity transformation of A e}
\cA^d=\cI\cA^e\cI^\top
\end{equation}
where $\cI$ is the permutation matrix defined by~\eqref{eq: permutation matrix block}. Using~\eqref{eq: jordan canonical form of A}, we can rewrite the second block in~\eqref{eq: block diagonal form of A e} as:
\begin{equation}
\label{eq: jordan canonical form of A conjugate}
(\cA^*)^\top=(\cV_{\epsilon}^*)^\top(\Lambda_{\epsilon}^*)^\top((\cV_{\epsilon}^{-1})^*)^\top=(\cV_{\epsilon}^*)^\top(\Lambda_{\epsilon}^*)^\top((\cV_{\epsilon}^*)^\top)^{-1}
\end{equation}
where $(\Lambda_{\epsilon}^*)^\top=\diag\{I_P,(\cJ_{\epsilon}^*)^\top\}$ and
\begin{equation}
(\cV_{\epsilon}^*)^\top=\left[
\begin{array}{cc}
(\cU^*)^\top&(\cV_{R,\epsilon}^*)^\top
\end{array}
\right],\quad\quad((\cV_{\epsilon}^{-1})^*)^\top=\left[
\begin{array}{c}
\cU^\top\\
\cV_{L,\epsilon}^\top
\end{array}
\right].
\end{equation}
Now, by replacing~\eqref{eq: jordan canonical form of A} and~\eqref{eq: jordan canonical form of A conjugate} into~\eqref{eq: block diagonal form of A e}, and by introducing the extended $2\times 2$ block diagonal matrices:
\begin{align}
\cV_{\epsilon}^d&=\diag\left\{\cV_{\epsilon},(\cV_{\epsilon}^*)^\top\right\},\label{eq: block diagonal matrix of eigenvectors}\\
\Lambda_{\epsilon}^d&=\diag\left\{\Lambda_{\epsilon},(\Lambda_{\epsilon}^*)^\top\right\},
\end{align}
we find that the $2M\times 2M$ matrix $\cA^d$ has a Jordan decomposition of the form:
\begin{equation}
\label{eq: eigendecomposition of Ad}
\cA^d=\cV^d_{\epsilon}\Lambda_{\epsilon}^d(\cV^d_{\epsilon})^{-1}. 
\end{equation}
Let us again introduce a permutation matrix $\cI'$ given by:
\begin{equation}
\label{eq: second permutation matrix}
\cI'\triangleq\left[
\begin{array}{cccc}
I_{P}&0&0&0\\
0&0&I_{P}&0\\
0&I_{M-P}&0&0\\
0&0&0&I_{M-P}
\end{array}
\right].
\end{equation} 
The matrix $\cA^d$ in~\eqref{eq: eigendecomposition of Ad} can be written alternatively as:
\begin{equation}
\cA^d=\underbrace{\cV^d_{\epsilon}\cI'^\top}_{\triangleq\cV^{d'}_{\epsilon}}\underbrace{\cI'\Lambda_{\epsilon}^d\cI'^\top}_{\triangleq\Lambda^e_{\epsilon}}\underbrace{\cI'(\cV^d_{\epsilon})^{-1}}_{(\cV^{d'}_{\epsilon})^{-1}} \label{eq: eigendecomposition of Ad 1}
\end{equation}
where the matrix $\Lambda_{\epsilon}^e$ is block diagonal defined in~\eqref{eq: partitioning of cVe-1}. Returning now to $\cA^e$, and using~\eqref{eq: similarity transformation of A e} and~\eqref{eq: eigendecomposition of Ad 1}, we find that the matrix $\cA^e$ has a Jordan decomposition of the form:
\begin{equation}
\cA^e=\cV^e_{\epsilon}\Lambda_{\epsilon}^e(\cV^e_{\epsilon})^{-1},
\end{equation}
where $\cV^e_{\epsilon}$ and $(\cV^e_{\epsilon})^{-1}$ are defined by:
\begin{equation}
\cV^e_{\epsilon}\triangleq\cI^\top\cV_{\epsilon}^d\cI'^\top,\qquad(\cV^e_{\epsilon})^{-1}\triangleq\cI'(\cV_{\epsilon}^d)^{-1}\cI
\end{equation}
in terms of the permutation matrices $\cI$ and $\cI'$ in~\eqref{eq: first permutation matrix},~\eqref{eq: second permutation matrix} and the block diagonal matrix $\cV_{\epsilon}^d$ in~\eqref{eq: block diagonal matrix of eigenvectors}. By properly evaluating these matrices, we arrive at~\eqref{eq: partitioning of cVe-1}.


\section{Proof of Theorem~\ref{theo: Network mean-square-error stability}}
\label{app: mean-square error stability}
We consider the transformed variables $\bcwb_{i}^e$ and $\bcwc_{i}^e$ in~\eqref{eq: transformed version of wt}. Conditioning both sides on $\bcF_{i-1}$, computing the conditional second-order moments, using the conditions from Assumption~\ref{assump: gradient noise} on the gradient noise process, and computing the expectations again, we get:
\begin{equation}
\begin{split}
\label{eq: variance of wci}
\expec\left\|\bcwb_{i}^e\right\|^2=\expec&\left\|(I_{hP}-\bcD_{11,i-1})\bcwb_{i-1}^e-\bcD_{12,i-1}\bcwc_{i-1}^e\right\|^2+\expec\left\|\bsb_{i}^e\right\|^2
\end{split}
\end{equation}
and
\begin{equation}
\label{eq: variance of wri}
\begin{split}
\expec\left\|\bcwc_{i}^e\right\|^2=\expec&\left\|(\cJ_{\epsilon}^e-\bcD_{22,i-1})\bcwc_{i-1}^e-\bcD_{21,i-1}\bcwb_{i-1}^e+\widecheck{b}^e\right\|^2+\expec\left\|\bsc_{i}^e\right\|^2.
\end{split}
\end{equation}
By applying Jensen's inequality to the convex function $\|x\|^2$, we can bound the first term on the RHS of~\eqref{eq: variance of wci} as follows:
\begin{align}
\expec\left\|(I_{hP}-\bcD_{11,i-1})\bcwb_{i-1}^e-\bcD_{12,i-1}\bcwc_{i-1}^e\right\|^2&=\expec\left\|(1-t)\frac{1}{1-t}(I_{hP}-\bcD_{11,i-1})\bcwb_{i-1}^e\hspace{-0.5mm}-t\frac{1}{t}\bcD_{12,i-1}\bcwc_{i-1}^e\right\|^2\notag\\
&\leq\frac{1}{1-t}\expec\left\|(I_{hP}-\bcD_{11,i-1})\bcwb_{i-1}^e\right\|^2+\frac{1}{t}\expec\left\|\bcD_{12,i-1}\bcwc_{i-1}^e\right\|^2\notag\\
&\leq\frac{1}{1-t}\expec\left[\left\|I_{hP}-\bcD_{11,i-1}\right\|^2\left\|\bcwb_{i-1}^e\right\|^2\right]+\frac{1}{t}\expec\left[\left\|\bcD_{12,i-1}\right\|^2\left\|\bcwc_{i-1}^e\right\|^2\right]\label{eq: first term on the RHS}
\end{align}
for any arbitrary positive number $t\in(0,1)$. By Assumption~\ref{assump: risks}, the Hermitian matrix $\bH_{k,i-1}$ defined in~\eqref{eq: Hki-1} can be bounded as follows:
\begin{equation}
\label{eq: bound on Hki-1}
0\leq\frac{\nu_k}{h}I_{hM_k}\leq\bH_{k,i-1}\leq\frac{\delta_k}{h}I_{hM_k}.
\end{equation}
Using the fact that the integral of a matrix is the matrix of the integrals, and  the linear property of integration, the Hermitian block $\bcD_{11,i-1}$ in~\eqref{eq: D11} can be rewritten as:
\begin{equation}
\bcD_{11,i-1}\hspace{-1mm}=\hspace{-0.5mm}\mu\int_{0}^1\hspace{-1mm}\left[(\cU^e)^*\diag\left\{\nabla_{w_k}^2J_k(w^o_k-t\bwt_{k,i-1})\right\}_{k=1}^N\cU^e\right] \hspace{-1mm}dt
\end{equation}
and, therefore, from Assumption~\ref{assump: risks}, $\bcD_{11,i-1}$ can be bounded as follows: 
\begin{equation}
\label{eq: bound on D_{11,i-1}}
0<\mu\frac{\nu}{h}I_{hP}\leq\bcD_{11,i-1}\leq\mu\frac{\delta}{h}I_{hP},
\end{equation}
for some positive constants $\nu$ and $\delta$ that are independent of $\mu$ and $i$. In terms of the $2-$induced matrix norm (i.e., maximum singular value), we obtain:
\begin{align}
\|I_{hP}-\bcD_{11,i-1}\|&=\rho(I_{hP}-\bcD_{11,i-1})\leq\max\left\{\left|1-\mu\frac{\delta}{h}\right|,\left|1-\mu\frac{\nu}{h}\right|\right\},
\end{align}
and, therefore,
\begin{equation}
\|I_{hP}-\bcD_{11,i-1}\|^2\leq(1-\mu\sigma_{11})^2,\label{eq: bound on I-D11 2}
\end{equation}
for some positive constant $\sigma_{11}$ that is independent of $\mu$ and $i$. 

Similarly, using the $2-$induced matrix norm (i.e., maximum singular value),  we can bound $\|\bcD_{12,i-1}\|^2$ as follows:
\begin{eqnarray}
\|\bcD_{12,i-1}\|^2&\overset{\eqref{eq: D12}}{=}&\|\mu\,(\cU^e)^*\bcH_{i-1}\cV_{R,\epsilon}^e\|^2\notag\\
&\leq&\mu^2\|(\cU^e)^*\|^2\|\bcH_{i-1}\|^2\|\cV_{R,\epsilon}^e\|^2\notag\\
&\leq&\mu^2\left(\max_{1\leq k\leq N}\|\bH_{k,i-1}\|^2\right)\|\cV_{R,\epsilon}^e\|^2\notag\\
&\overset{\eqref{eq: bound on Hki-1}}{\leq}&\mu^2\|\cV_{R,\epsilon}^e\|^2\max_{1\leq k\leq N}\left\{\frac{\delta_k^2}{h^2}\right\}\notag\\
&=&\mu^2\sigma_{12}^2,\label{eq: bound on D12}
\end{eqnarray}
for some \cred{positive} constant $\sigma_{12}$ and where we used the fact that $\|(\cU^e)^*\|=\sigma_{\max}((\cU^e)^*)=\sqrt{\lambda_{\max}(\cU^e(\cU^e)^*)}=1$.

Substituting~\eqref{eq: first term on the RHS} into~\eqref{eq: variance of wci}, and using~\eqref{eq: bound on I-D11 2},~\eqref{eq: bound on D12}, we get:
\begin{equation}
\label{eq: variance of wci 1}
\expec\|\bcwb_{i}^e\|^2\leq\frac{(1-\sigma_{11}\mu)^{2}}{1-t}\expec\|\bcwb_{i-1}^e\|^2+\frac{\mu^2\sigma_{12}^2}{t}\expec\|\bcwc_{i-1}^e\|^2+\expec\|\bsb_{i}^e\|^2
\end{equation}
We select $t=\sigma_{11}\mu$ (for sufficiently small $\mu$). Then, the previous inequality can be written as:
\begin{equation}
\label{eq: variance of wci 2}
\expec\|\bcwb_{i}^e\|^2\leq(1-\sigma_{11}\mu)\expec\|\bcwb_{i-1}^e\|^2+\frac{\mu\sigma_{12}^2}{\sigma_{11}}\expec\|\bcwc_{i-1}^e\|^2+\expec\|\bsb_{i}^e\|^2.
\end{equation}

We  repeat  similar arguments for the second variance relation~\eqref{eq: variance of wri}. Using Jensen's inequality again, \cmag{we obtain:
\begin{equation}
\label{eq: eq: variance of wri 3}
\begin{split}
&\expec\|(\cJ_{\epsilon}^e-\bcD_{22,i-1})\bcwc_{i-1}^e-\bcD_{21,i-1}\bcwb_{i-1}^e+\widecheck{b}^e\|^2\\
&\leq\frac{1}{t}\expec\|\cJ_{\epsilon}^e\bcwc_{i-1}^e\|^2+\frac{1}{1-t}\expec\|\bcD_{22,i-1}\bcwc_{i-1}^e+\bcD_{21,i-1}\bcwb_{i-1}^e-\widecheck{b}^e\|^2\\
&\overset{\text{(a)}}{\leq}\frac{(\rho(\cJ_{\epsilon})+\epsilon)^2}{t}\expec\|\bcwc^e_{i-1}\|^2+\frac{1}{1-t}\expec\|\bcD_{22,i-1}\bcwc^e_{i-1}+\bcD_{21,i-1}\bcwb^e_{i-1}-\widecheck{b}^e\|^2\\
&\overset{\text{(b)}}= (\rho(\cJ_{\epsilon})+\epsilon)\expec\|\bcwc^e_{i-1}\|^2+\frac{1}{1-\rho(\cJ_{\epsilon})-\epsilon}\expec\|\bcD_{22,i-1}\bcwc^e_{i-1}+\bcD_{21,i-1}\bcwb^e_{i-1}-\widecheck{b}^e\|^2,
\end{split}
\end{equation}
for any arbitrary positive number $t\in(0,1)$. In (a) we used the fact that the block diagonal matrix $\cJ^e_\epsilon$ defined in \cmag{Table~\ref{table: variables table}} satisfies:
\begin{equation}
\label{eq: norm 2 J 1}
 \|\cJ_{\epsilon}^e\|^2\leq(\rho(\cJ_{\epsilon})+\epsilon)^2.
\end{equation}
Expression~\eqref{eq: norm 2 J 1} can be established by using similar arguments as in~\cite[pp.~516]{sayed2014adaptation} and the fact that 
$\lambda_{\max}\left(\cJ_\epsilon^\top(\cJ_\epsilon^*)^\top\right)=\lambda_{\max}\left((\cJ_\epsilon^*\cJ_\epsilon)^\top\right)=\lambda_{\max}(\cJ_\epsilon^*\cJ_\epsilon)$.
In (b), we used the fact that $\rho(\cJ_{\epsilon})\in{(0},1)$, and thus, $\epsilon$  can be selected small enough to ensure $\rho(\cJ_{\epsilon})+\epsilon\in(0,1)$.  We then selected $t=\rho(\cJ_{\epsilon})+\epsilon$.
}
Using Jensen's inequality, the second term on the RHS of~\eqref{eq: eq: variance of wri 3} can be bounded as follows:
\begin{equation}
\label{eq: eq: variance of wri 4}
\begin{split}
&\expec\|\bcD_{22,i-1}\bcwc^e_{i-1}+\bcD_{21,i-1}\bcwb^e_{i-1}-\widecheck{b}^e\|^2\leq 3\expec[\|\bcD_{22,i-1}\|^2\|\bcwc^e_{i-1}\|^2]+3\expec[\| \bcD_{21,i-1}\|^2\|\bcwb^e_{i-1}\|^2]+3\|\widecheck{b}^e\|^2
\end{split}
\end{equation}
Following similar arguments as in~\eqref{eq: bound on D12}, we can show that: 
\begin{align}
\|\bcD_{21,i-1}\|^2
\leq\mu^2\sigma_{21}^2,\qquad \|\bcD_{22,i-1}\|^2
\leq\mu^2\sigma_{22}^2,\label{eq: bound on norm D22}
\end{align}
for some \cred{positive} constants $\sigma_{21}$ and $\sigma_{22}$. \cmag{Substituting~\eqref{eq: eq: variance of wri 4} into~\eqref{eq: eq: variance of wri 3} and~\eqref{eq: eq: variance of wri 3} into~\eqref{eq: variance of wri}, and using~\eqref{eq: bound on norm D22}, we obtain:}
\begin{equation}
\label{eq: variance of wri 6}
\begin{split}
&\expec\|\bcwc_{i}^e\|^2\\
&\leq\left(\rho(\cJ_{\epsilon})+\epsilon+\frac{3\mu^2\sigma_{22}^2}{1-\rho(\cJ_{\epsilon})-\epsilon}\right)\expec\|\bcwc_{i-1}^e\|^2+\left(\frac{3\mu^2\sigma_{21}^2}{1-\rho(\cJ_{\epsilon})-\epsilon}\right)\expec\|\bcwb^e_{i-1}\|^2+\left(\frac{3}{1-\rho(\cJ_{\epsilon})-\epsilon}\right)\|\widecheck{b}^e\|^2+\expec\|\bsc_{i}^e\|^2.
\end{split}
\end{equation}
From~\eqref{eq: transformed version of the bias}, we have:
\begin{equation}
\label{eq: norm of br}
\|\widecheck{b}^e\|^2=\|\mu\,\cJ_{\epsilon}^e(\cV_{L,\epsilon}^{e})^*b^e\|^2\leq\mu^2(\rho(\cJ_{\epsilon})+\epsilon)^2\|(\cV_{L,\epsilon}^{e})^*\|^2\|b^e\|^2
\end{equation}
where we used the fact that $(\cV_{L,\epsilon}^{e})^*\cA^e=\cJ_{\epsilon}^e(\cV_{L,\epsilon}^{e})^*$ from Lemma~\ref{lemm: jordan decomposition}. Since $b^e$ in~\eqref{eq: expression of b 1} is defined in terms of the gradient $\nabla_{w^*_k}J_k(w^o_k)$ and since $J_k(w_k)$ is twice differentiable, then $\|b^e\|^2$ is bounded and we obtain $\|\widecheck{b}^e\|^2=O(\mu^2)$. For the noise terms  $\expec\|\bsb_{i}^e\|^2$ in~\eqref{eq: variance of wci 2} and $\expec\|\bsc_{i}^e\|^2$ in~\eqref{eq: variance of wri 6}, we have:
\begin{equation}
\label{eq: sum of gradient noise variances}
\expec\|\bsb_{i}^e\|^2+\expec\|\bsc_{i}^e\|^2\overset{\eqref{eq: transformed version of the noise}}=\expec\|\mu(\cV_{\epsilon}^e)^{-1}\cA^e\bs_i^e\|^2\leq\mu^2 v_1^2\expec\|\bs_i^e\|^2,
\end{equation}
where $v_1$ is a positive constant independent of $\mu$ and given by $v_1\triangleq\|(\cV_{\epsilon}^e)^{-1}\cA^e\|=\|\Lambda_\epsilon^e(\cV_{\epsilon}^e)^{-1}\|$. In terms of the variances of the individual noise processes, $\expec\|\bs_{k,i}\|^2$, we have 
$\expec\|\bs_i^e\|^2=\sum_{k=1}^N\expec\|\bs^e_{k,i}\|^2=2\sum_{k=1}^N\expec\|\bs_{k,i}\|^2$.
For each   $\bs_{k,i}(\bw_{k,i-1})$, we have from Assumption~\ref{assump: gradient noise} and Jensen's inequality:
\begin{align}
\expec\|\bs_{k,i}(\bw_{k,i-1})\|^2&\leq (\beta_k^2/h^2)\expec\|w^o_k-\bw_{k,i-1}+w^o_k\|^2+\sigma^2_{s,k}\notag\\
&\leq2(\beta_k^2/h^2)\expec\|\bwt_{k,i-1}\|^2+2(\beta_k^2/h^2)\|w^o_k\|^2+\sigma^2_{s,k}\notag\\
&=\bar{\beta}_k^2\expec\|\bwt_{k,i-1}\|^2+\bar{\sigma}_{s,k}^2\label{eq: second order moment of gradient noise and error}
\end{align}
where $\bar{\beta}_k^2\triangleq2(\beta_k^2/h^2)$ and $\bar{\sigma}_{s,k}^2\triangleq2(\beta_k^2/h^2)\|w^o_k\|^2+\sigma^2_{s,k}$. The term $\expec\|\bs_i^e\|^2$ can thus be bounded as follows:
\begin{align}
\expec\|\bs_i\|^2&\leq2\sum_{k=1}^N\bar{\beta}_k^2\expec\|\bwt_{k,i-1}\|^2+2\sum_{k=1}^N\bar{\sigma}_{s,k}^2\notag\\
&\leq \beta^2_{\max}\expec\|\cV_{\epsilon}^e(\cV^e_{\epsilon})^{-1}\bcwt^e_{i-1}\|^2+\sigma^2_s\notag\\
&\leq\beta^2_{\max}\|\cV^e_{\epsilon}\|^2\expec\|(\cV^e_{\epsilon})^{-1}\bcwt^e_{i-1}\|^2+\sigma^2_s\notag\\
&=\beta^2_{\max}v_2^2\left[\expec\|\bcwb_{i-1}^e\|^2+\expec\|\bcwc_{i-1}^e\|^2\right]+\sigma^2_s
\end{align}
where $\beta^2_{\max}\triangleq\max_{1\leq k\leq N}\bar{\beta}_k^2$, $\sigma^2_s\triangleq2\sum_{k=1}^N\bar{\sigma}_{s,k}^2$, and $v_2\triangleq\|\cV_{\epsilon}^e\|$. Substituting into~\eqref{eq: sum of gradient noise variances}, we get:
\begin{equation}
\label{eq: sum of gradient noise variances 1}
\begin{split}
&\expec\|\bsb_{i}^e\|^2+\expec\|\bsc_{i}^e\|^2\leq\mu^2 v_1^2\beta^2_{\max}v_2^2[\expec\|\bcwb^e_{i-1}\|^2+\expec\|\bcwc^e_{i-1}\|^2]+\mu^2 v_1^2\sigma^2_s.
\end{split}
\end{equation}
Using this bound in~\eqref{eq: variance of wci 2} and~\eqref{eq: variance of wri 6}, we obtain:
\begin{equation}
\label{eq: variance of wci 3}
\begin{split}
\expec\|\bcwb^e_{i}\|^2\leq&(1-\sigma_{11}\mu+\mu^2 v_1^2\beta^2_{\max}v_2^2)\expec\|\bcwb^e_{i-1}\|^2+\left(\frac{\mu\sigma_{12}^2}{\sigma_{11}}+\mu^2 v_1^2\beta^2_{\max}v_2^2\right)\expec\|\bcwc^e_{i-1}\|^2+\mu^2 v_1^2\sigma^2_s,
\end{split}
\end{equation}
\begin{equation}
\label{eq: variance of wri 7}
\begin{split}
\expec\|\bcwc^e_{i}\|^2\leq&\left(\frac{3\mu^2\sigma_{21}^2}{1-\rho(\cJ_{\epsilon})-\epsilon}+\mu^2 v_1^2\beta^2_{\max}v_2^2\right)\expec\|\bcwb^e_{i-1}\|^2+\left(\rho(\cJ_{\epsilon})+\epsilon+\frac{3\mu^2\sigma_{22}^2}{1-\rho(\cJ_{\epsilon})-\epsilon}+\mu^2 v_1^2\beta^2_{\max}v_2^2\right)\expec\|\bcwc^e_{i-1}\|^2\\
&+\left(\frac{3}{1-\rho(\cJ_{\epsilon})-\epsilon}\right)\|\widecheck{b}^e\|^2+\mu^2 v_1^2\sigma^2_s.
\end{split}
\end{equation}
We can combine~\eqref{eq: variance of wci 3} and~\eqref{eq: variance of wri 7} into a single inequality recursion:
\begin{equation}
\label{eq: error variance recursion}
\left[
\begin{array}{c}
\expec\|\bcwb^e_{i}\|^2\\
\expec\|\bcwc^e_{i}\|^2
\end{array}
\right]\preceq
\Gamma
\left[
\begin{array}{c}
\expec\|\bcwb^e_{i-1}\|^2\\
\expec\|\bcwc^e_{i-1}\|^2
\end{array}
\right]+
\left[
\begin{array}{c}
e\\
e+f
\end{array}
\right]
\end{equation}
where $\Gamma$ is given by:
\begin{equation}
\Gamma\triangleq\left[
\begin{array}{cc}
a&b\\
c&d
\end{array}
\right]=\left[
\begin{array}{cc}
1-O(\mu)&O(\mu)\\
O(\mu^2)&\rho(\cJ_{\epsilon})+\epsilon+O(\mu^2)
\end{array}
\right].
\end{equation}
and where $a=1-O(\mu)$, $b=O(\mu)$, $c=O(\mu^2)$, $d=\rho(\cJ_{\epsilon})+\epsilon+O(\mu^2)$, $e=O(\mu^2)$, and $f=O(\mu^2)$. Now, using the property that the spectral radius of a matrix is upper bounded by its  $1-$norm norm, we obtain:
\begin{equation}
\label{eq: stability of gamma condition}
\rho(\Gamma)\leq\max\left\{1-O(\mu)+O(\mu^2),\rho(\cJ_{\epsilon})+\epsilon+O(\mu)+O(\mu^2)\right\}
\end{equation}
Since $\rho(\cJ_{\epsilon})<1$ is independent of $\mu$, and since $\epsilon$ and $\mu$ are small positive numbers that can be chosen arbitrarily small and independently of each other, it is clear that the RHS of the above expression can be made strictly smaller than one for sufficiently small $\epsilon$ and $\mu$. In that case $\rho(\Gamma)<1$ so that $\Gamma$ is stable. Moreover, it holds that:
\begin{align}
\label{eq: inverse of I-Gamma}
(I-\Gamma)^{-1}&=\left[
\begin{array}{cc}
O(1/\mu)&O(1)\\
O(\mu)&O(1)
\end{array}
\right].
\end{align}
Now, by iterating~\eqref{eq: error variance recursion} we arrive at:
\begin{align}
\limsup_{i\rightarrow\infty}\left[
\begin{array}{c}
\expec\|\bcwb^e_{i}\|^2\\
\expec\|\bcwc^e_{i}\|^2
\end{array}
\right]&\preceq(I-\Gamma)^{-1}\left[
\begin{array}{c}
e\\
e+f
\end{array}
\right]
=\left[
\begin{array}{c}
O(\mu)\\
O(\mu^2)
\end{array}
\right]
\end{align}
from which we conclude that $\limsup_{i\rightarrow\infty}\expec\|\bcwb^e_{i}\|^2=O(\mu)$ and $\limsup_{i\rightarrow\infty}\expec\|\bcwc^e_{i}\|^2=O(\mu^2)$. Therefore,
\begin{align}
\limsup_{i\rightarrow\infty}\expec\|\bcwt^e_{i}\|^2&=\limsup_{i\rightarrow\infty}\expec\left\|\cV_{\epsilon}^e\left[\begin{array}{c}\bcwb^e_{i}\\\bcwc^e_{i}\end{array}\right]\right\|^2\notag\\
&\leq\limsup_{i\rightarrow\infty}v_2^2\left[\expec\|\bcwb^e_{i}\|^2+\expec\|\bcwc^e_{i}\|^2\right]=O(\mu).\label{eq: final result in the proof of mean-square stability}
\end{align}

\end{appendices}

\bibliographystyle{IEEEbib}
{\balance{
\bibliography{reference}}}

\end{document}